\documentclass[11pt]{article}

\usepackage{amsthm}
\usepackage{graphicx} 
\usepackage{array} 

\usepackage{amsmath, amssymb, amsfonts, verbatim}
\usepackage{hyphenat,epsfig,subcaption,multirow}
\usepackage[font=small,labelfont=bf]{caption}
\usepackage{nicefrac}
\usepackage{paralist}

\usepackage[usenames,dvipsnames]{xcolor}
\usepackage[ruled]{algorithm2e}

\DeclareFontFamily{U}{mathx}{\hyphenchar\font45}
\DeclareFontShape{U}{mathx}{m}{n}{
      <5> <6> <7> <8> <9> <10>
      <10.95> <12> <14.4> <17.28> <20.74> <24.88>
      mathx10
      }{}
\DeclareSymbolFont{mathx}{U}{mathx}{m}{n}
\DeclareMathSymbol{\bigtimes}{1}{mathx}{"91}

\usepackage{tcolorbox}
\tcbuselibrary{skins,breakable}
\tcbset{enhanced jigsaw}

\usepackage[normalem]{ulem}
\usepackage[compact]{titlesec}

\definecolor{DarkRed}{rgb}{0.5,0.1,0.1}
\definecolor{DarkBlue}{rgb}{0.1,0.1,0.5}

\usepackage{nameref}
\definecolor{ForestGreen}{rgb}{0.1333,0.5451,0.1333}
\definecolor{Red}{rgb}{0.9,0,0}
\usepackage[linktocpage=true,
	pagebackref=true,colorlinks,
	linkcolor=DarkRed,citecolor=ForestGreen,
	bookmarks,bookmarksopen,bookmarksnumbered]
	{hyperref}
\usepackage[noabbrev,nameinlink]{cleveref}
\crefname{property}{property}{Property}
\creflabelformat{property}{(#1)#2#3}
\crefname{equation}{eq}{Eq}
\creflabelformat{equation}{(#1)#2#3}

\usepackage{bm}
\usepackage{url}
\usepackage{xspace}
\usepackage[mathscr]{euscript}

\usepackage{tikz}
\usetikzlibrary{arrows}
\usetikzlibrary{arrows.meta}
\usetikzlibrary{shapes}
\usetikzlibrary{backgrounds}
\usetikzlibrary{positioning}
\usetikzlibrary{decorations.markings}
\usetikzlibrary{patterns}
\usetikzlibrary{calc}
\usetikzlibrary{fit}
\usetikzlibrary{decorations}

\usepackage[framemethod=TikZ]{mdframed}

\usepackage[noend]{algpseudocode}
\makeatletter
\def\BState{\State\hskip-\ALG@thistlm}
\makeatother

\usepackage{cite}
\usepackage{enumitem}
\setlist[itemize]{leftmargin=20pt}
\setlist[enumerate]{leftmargin=20pt}

\usepackage[margin=1in]{geometry}

\usepackage{thmtools}
\usepackage{thm-restate}

\newtheorem{theorem}{Theorem}
\newtheorem{lemma}{Lemma}[section]
\newtheorem{proposition}[lemma]{Proposition}

\newtheorem{claim}[lemma]{Claim}
\newtheorem{fact}[lemma]{Fact}

\newtheorem{definition}[lemma]{Definition}
\newtheorem{problem}{Problem}

\newtheorem*{claim*}{Claim}
\newtheorem*{assumption*}{Assumption}
\newtheorem*{proposition*}{Proposition}
\newtheorem*{lemma*}{Lemma}

\newtheorem*{theorem*}{Theorem}

\crefname{lemma}{Lemma}{Lemmas}
\crefname{claim}{claim}{claims}
\crefname{property}{Property}{Properties}
\crefname{invariant}{Invariant}{Invariants}

\newtheorem{mdresult}{Result}
\newenvironment{result}{\begin{mdframed}[backgroundcolor=lightgray!40,topline=false,rightline=false,leftline=false,bottomline=false,innertopmargin=2pt]\begin{mdresult}}{\end{mdresult}\end{mdframed}}

\theoremstyle{definition}

\newenvironment{Problem}{\begin{mdframed}[backgroundcolor=ForestGreen!15,topline=false,bottomline=false, innerbottommargin=12pt,innertopmargin=12pt]\begin{problem}}{\end{problem}\end{mdframed}}
\newtheorem*{mdproblem*}{Problem}
\newenvironment{Problem*}{\begin{mdframed}[hidealllines=false,innerleftmargin=10pt,backgroundcolor=gray!10,innertopmargin=5pt,innerbottommargin=5pt,roundcorner=10pt]\begin{mdproblem*}}{\end{mdproblem*}\end{mdframed}}
\newtheorem{mddefinition}[lemma]{Definition}

\newtheorem*{mddefinition*}{Definition}
\newenvironment{Definition*}{\begin{mdframed}[hidealllines=false,innerleftmargin=10pt,backgroundcolor=white!10,innertopmargin=5pt,innerbottommargin=5pt,roundcorner=10pt]\begin{mddefinition*}}{\end{mddefinition*}\end{mdframed}}
\newtheorem{mdremark}{Remark}

\newenvironment{ourbox}{\begin{mdframed}[hidealllines=false,innerleftmargin=10pt,backgroundcolor=white!10,innertopmargin=2pt,innerbottommargin=5pt,roundcorner=10pt]}{\end{mdframed}}

\newtheorem{mdalgorithm}{Algorithm}
\newenvironment{Algorithm}{\begin{ourbox}\begin{mdalgorithm}}{\end{mdalgorithm}\end{ourbox}}

\allowdisplaybreaks

\DeclareMathOperator*{\argmax}{arg\,max}

\renewcommand{\qed}{\nobreak \ifvmode \relax \else
      \ifdim\lastskip<1.5em \hskip-\lastskip
      \hskip1.5em plus0em minus0.5em \fi \nobreak
      \vrule height0.75em width0.5em depth0.25em\fi}

\renewcommand{\leq}{\leqslant}
\renewcommand{\geq}{\geqslant}

\renewcommand{\ge}{\geq}

\newcommand{\rs}{{{Ruzsa-Szemerédi}}\xspace}

\newcommand{\Leq}[1]{\ensuremath{\underset{\textnormal{#1}}\leq}}

\newcommand{\Ot}{\ensuremath{\widetilde{O}}}
\newcommand{\eps}{\ensuremath{\varepsilon}}
\newcommand{\Paren}[1]{\Big(#1\Big)}

\newcommand{\bracket}[1]{\left[#1\right]}
\newcommand{\paren}[1]{\ensuremath{\left(#1\right)}\xspace}
\newcommand{\card}[1]{\left\vert{#1}\right\vert}

\newcommand{\expect}[1]{\Exp\bracket{#1}}

\newcommand{\set}[1]{\ensuremath{\left\{ #1 \right\}}}
\newcommand{\poly}{\mbox{\rm poly}}

\newcommand{\alg}{\ensuremath{\mathcal{A}}\xspace}

\DeclareMathOperator*{\Exp}{\ensuremath{{\mathbb{E}}}}
\DeclareMathOperator*{\Prob}{\ensuremath{\textnormal{Pr}}}
\renewcommand{\Pr}{\Prob}

\newenvironment{tbox}{\begin{tcolorbox}[
		enlarge top by=5pt,
		enlarge bottom by=5pt,
		 breakable,
		 boxsep=0pt,
                  left=4pt,
                  right=4pt,
                  top=10pt,
                  arc=0pt,
                  boxrule=1pt,toprule=1pt,
                  colback=white
                  ]
	}
{\end{tcolorbox}}

\newcommand{\II}{\ensuremath{\mathbb{I}}}

\newcommand{\mireal}[1][]{
  \ifx\relax#1\relax%
    \II(\mione \,; \mitwo)%
  \else%
    \II(\mione \,; \mitwo\mid #1)%
  \fi
}

\title{Improved Bounds for Fully Dynamic Matching via \\ Ordered \rs Graphs} 
\author{Sepehr Assadi%
\footnote{(sepehr@assadi.info) Cheriton School of Computer Science, University of Waterloo. 
Supported in part by a  Sloan Research Fellowship, an NSERC
Discovery Grant, a University of Waterloo startup grant, and a Faculty of Math Research Chair grant. \smallskip} \and
Sanjeev Khanna%
\footnote{(sanjeev@cis.upenn.edu) Department of Computer and Information Science, University of Pennsylvania. Research supported in part by NSF awards CCF-1934876, CCF-2008305, and CCF-2402284.\smallskip} \and 
Peter Kiss%
\footnote{(peter.kiss@warwick.ac.uk) University of Vienna. The author completed parts of this project at the University of Warwick.}
}

\date{}

\begin{document}
\maketitle

\pagenumbering{roman}


\begin{abstract}

In a very recent breakthrough, Behnezhad and Ghafari [FOCS'24] developed a novel fully dynamic randomized algorithm for maintaining a $(1-\eps)$-approximation of maximum matching 
with amortized update time \emph{potentially} much better than the trivial $O(n)$ update time. The runtime of the BG algorithm is parameterized 
via the following graph theoretical concept: 
\begin{itemize}
	\item For any $n$, define $\mathsf{ORS}(n)$---standing for \emph{Ordered \rs Graph}---to be the largest number of edge-disjoint matchings $M_1,\ldots,M_t$ of size $\Theta(n)$ in an $n$-vertex graph 
such that for every $i \in [t]$, $M_i$ is an induced matching in the subgraph $M_{i} \cup M_{i+1} \cup \ldots \cup M_t$.
\end{itemize}
Then, for any fixed $\eps > 0$, the BG algorithm runs in 
\[
	O\left( \sqrt{n^{1+O(\eps)} \cdot \mathsf{ORS}(n)} \right)
\]
amortized update time with high probability, even against an adaptive adversary. $\mathsf{ORS}(n)$ is a close variant of a more well-known quantity regarding \rs graphs (which require every matching to be induced 
regardless of the ordering). It is currently only known that $n^{o(1)} \leq \mathsf{ORS}(n) \leq n^{1-o(1)}$, and closing this gap appears to be a 
notoriously challenging problem. 

\medskip

\emph{If} it turns out that $\mathsf{ORS}(n) = n^{o(1)}$, namely, the current lower bounds are close to being optimal, then, this algorithm achieves an update time of $n^{1/2+o(1)}$ for
$(1-\eps)$-approximation of fully dynamic matching, making progress on a major open question in the area. 

\medskip

\paragraph{Our Result:} In this work, we further strengthen the result of Behnezhad and Ghafari and push it to limit to obtain a randomized algorithm with amortized update time of
\[
	n^{o(1)} \cdot \mathsf{ORS}(n)
\]
with high probability, even against an adaptive adversary.
In the limit, i.e., \emph{if} current lower bounds for $\mathsf{ORS}(n) = n^{o(1)}$ are almost optimal, our algorithm achieves an $n^{o(1)}$ update time for $(1-\eps)$-approximation of maximum matching, almost fully resolving this fundamental question. In its current stage also, this fully reduces the algorithmic problem
of designing dynamic matching algorithms to a purely combinatorial problem of upper bounding $\mathsf{ORS}(n)$ with no algorithmic considerations.

\end{abstract}

\clearpage

\setcounter{tocdepth}{3}
\tableofcontents
\clearpage
\pagenumbering{arabic}
\setcounter{page}{1}


\newcommand{\ORS}[2]{\ensuremath{\textnormal{\textsf{ORS}}(#1,#2)}\xspace}

\renewcommand{\alg}{\textnormal{{\large \ensuremath{\mathbb{A}}}}}

\newcommand{\Algweak}{\ensuremath{\alg_{\textnormal{weak}}}\xspace}

\newcommand{\Algsub}{\ensuremath{\alg_{\textnormal{sublinear}}}\xspace}

\newcommand{\Algadd}{\ensuremath{\alg_{\textnormal{additive}}}\xspace}

\newcommand{\tmu}{\tilde{\mu}}

\newcommand{\density}[1]{\ensuremath{\textnormal{density}(#1)}\xspace}

\newcommand{\indeg}[1]{\ensuremath{{\Delta}_{\textsc{in}}(#1)}\xspace}

\newcommand{\MM}{\ensuremath{\mathcal{M}}}

\newcommand{\dstar}{d^{*}}

\newcommand{\MMstar}{\ensuremath{\MM^{*}}}

\newcommand{\degstar}{\ensuremath{\deg^*}}

\section{Introduction}\label{sec:intro}

We study the problem of maintaining an approximate maximum matching in a fully dynamic graph. In this problem, we have a graph $G=(V,E)$ that undergoes insertion and deletion of edges by an adversary 
and our goal is to maintain (edges of) an approximate maximum matching of $G$ after each update. This is one of the most central problems in the dynamic graph literature; see~\cite{OnakR10,GuptaP13,BernsteinS15,BernsteinS16,Solomon16,BaswanaGS18,BehnezhadK22,Behnezhad23,BhattacharyaKSW23,AssadiBKL23,BhattacharyaKS23,Liu24,BehnezhadG24} and references therein.

In a very recent breakthrough,~\cite{BehnezhadG24} developed an algorithm that for any fixed $\eps > 0$, maintains (edges of) a $(1-\eps)$-approximate maximum matching in a fully dynamic graph 
with \emph{potentially} much better than $O_{\eps}(n)$ update time. Specifically, the runtime of the algorithm of~\cite{BehnezhadG24} is parameterized based on the density of a certain family of extremal graphs which are (quite) closely related to \rs (RS) graphs~\cite{RuzsaS78} (see~\cite{AlonMS12,FoxHS17} for more
context on RS graphs, and~\cite{GoelKK12,AssadiBKL23,AssadiS23} and references therein for their applications to dynamic graph and other sublinear algorithms). 
\cite{BehnezhadG24} defined the following family of closely related graphs. 

\begin{definition}[Ordered \rs (ORS) Graphs~\cite{BehnezhadG24}]\label{def:ors}
	A graph $G=(V,E)$ is called an \textbf{\bm{$(r,t)$}\emph{-ORS graph}} if its edges can be partitioned into an \underline{ordered set} of $t$ matchings $M_1,\ldots,M_t$ each of size $r$, such that for every $i \in [t]$, 
	the matching $M_i$ is an induced matching in the subgraph of $G$ on $M_{i} \cup M_{i+1} \cup \ldots \cup M_t$. 

	We define $\ORS{n}{r}$ as the largest choice of $t$ such that an $n$-vertex $(r,t)$-ORS graph exist. 
\end{definition}

Unfortunately, exactly as in RS graphs, density of ORS graphs is quite poorly understood at this point. Currently, for any constant $\delta \in (0,1/4)$, it is only known that 
\begin{align}
	n^{\Omega_{\delta}(1/\log\log{n})} \Leq{\cite{FischerLNRRS02,GoelKK12}} \ORS{n}{\delta n} \Leq{\cite{BehnezhadG24}} \frac{n}{\log^{(\text{poly}(1/\delta))}{(n)}},  \label{eq:ORS}
\end{align}
where $\log^{(k)}(n)$ is the $k$-iterated logarithm function, i.e., 
\[
\log^{(k)}{(n)} := \underbrace{\log\log\cdots\log}_{k}{(n)}.
\]
This is quite similar to the situation for RS graphs (modulo a slightly better dependence in RS graphs on the parameter $\delta$ in the upper bound due to~\cite{Fox11} (see also~\cite{FoxHS17}), 
namely, $\log^{O(\log{(1/\delta)})}(n)$ instead in the denominator).  

The result of~\cite{BehnezhadG24} is a randomized algorithm that given any fixed $\eps > 0$ with high probability maintains a $(1-\eps)$-approximation of maximum matching in a fully dynamic graph with an amortized update time of
\[
	O\Paren{\sqrt{n^{1+\eps} \cdot \ORS{n}{\Theta_{\eps}(n)}}}.
\]
Thus, \emph{if} it happens to be the case that ORS graphs cannot be dense, i.e., $\ORS{n}{\Theta_{\eps}(n)} = n^{1-\Omega(1)}$, this algorithm achieves an update time of $n^{1-\Omega(1)}$ for this problem, 
making progress on a major open question in the dynamic matching literature~\cite{GuptaP13,BernsteinS16,BehnezhadK22,BhattacharyaKS23,BehnezhadG24}. Moreover, in the limit, namely, if the current lower bounds of~\Cref{eq:ORS} on ORS are almost optimal, then, this algorithm achieves an update time of $n^{1/2+o(1)}$; currently, the best algorithm known with such an update time due to~\cite{BernsteinS16} can only achieve a $2/3$-approximation.

\subsection{Our Contribution} 
We build on the approach of~\cite{BehnezhadG24} and push it to its limit to obtain the following result. 

\begin{result}\label{res:main}
	There is an algorithm that for any fixed $\eps > 0$ maintains a $(1-\eps)$-approximate maximum matching of any fully dynamic graph with amortized update time of
	\[
		O\Paren{n^{o(1)} \cdot \ORS{n}{\Theta_{\eps}(n)}}. \footnote{More specifically, for any $\beta_0 \in (0,1)$, there is a $\beta_1 \in (0,1)$ such that this runtime is $O(n^{\beta_0} \cdot \ORS{n}{\beta_1 \cdot f(\eps) \cdot n})$
		for some fixed function $f$ independent of $\beta_0,\beta_1$.} 
	\]
	The algorithm is randomized and its guarantees hold with high probability against an adaptive adversary. The algorithm does \underline{not} assume a prior knowledge of the value of $\ORS{n}{\Theta_{\eps}(n)}$
	to achieve its guarantee. 
\end{result}

In the limit, \emph{if} ORS graphs cannot be much denser than the lower bounds in~\Cref{eq:ORS},~\Cref{res:main} achieves an $n^{o(1)}$ amortized update time, almost fully settling the question of $(1-\eps)$-approximation of fully dynamic matching. 
Beside the conditional upper bound of~\cite{BehnezhadG24} (which would be an $n^{1/2+o(1)}$ update time algorithm under this hypothesis), it is also known unconditionally how to obtain an update time of $n/2^{\Theta(\sqrt{\log{n}})} \cdot \poly(1/\eps)$ on bipartite graphs~\cite{Liu24} (and presumably a similar runtime with $O(1/\eps)^{O(1/\eps)}$-dependence instead on general graphs using the reduction of~\cite{McGregor05}; see~\Cref{prop:boosting}). Moreover,~\cite{BhattacharyaKS23}, building on~\cite{Behnezhad23,BhattacharyaKSW23},   
designed an algorithm that for any fixed $\eps > 0$, obtains an update time of $m^{1/2-\Omega_{\eps}(1)}$ for the easier problem of maintaining the \emph{size} of the maximum matching (but not its edges).

\Cref{res:main} also suggests that almost any interesting lower bound for this problem (even under computational hardness assumptions) should effectively rule out existence of even mildly dense ORS graphs, which will  constitute a big breakthrough given the close connection of these graphs to RS graphs (or alternatively, explicitly condition on the assumption that ORS graphs are dense; similar assumptions for RS graphs have been used to prove lower bounds for $(1-\eps)$-approximation of the maximum matching problem in other settings, e.g., in the streaming model~\cite{AssadiS23}). 

It is worth noting here that existing conditional lower bounds in~\cite{HenzingerKNS15} rule out $n^{1-\Omega(1)}$ update time algorithms for computing an \emph{exact} maximum matching, 
and more recently in~\cite{Liu24}, for even a $(1-\eps)$-approximation but only when $\eps = n^{-\Omega(1)}$. On the other hand, our focus in~\Cref{res:main} is on the regime when $\eps \in (0,1)$ is fixed and independent of $n$ (which is often the main regime of 
interest in the context of sublinear algorithms). Indeed, the RS graph constructions of~\cite{RuzsaS78} (see also~\cite{AlonMS12}) imply that $\ORS{n}{\delta n} = \Omega(n)$ for $\delta \leq 2^{-\Theta(\sqrt{\log{n}})}$. This implies that 
for small enough $\eps$, it is already known that~\Cref{res:main} cannot achieve any non-trivial guarantee. 

\subsection{Our Algorithm at a High Level} 
By the existing boosting frameworks for matchings (see~\Cref{prop:boosting}), obtaining an $\eps n$ \emph{additive} approximation reduces
to the following problem: given a fully dynamic graph $G=(V,E)$, every $\Theta_{\eps}(n)$ updates we receive $O_\eps(1)$ \empty{queries} of the form $U \subseteq V$ and must return an $O(1)$-approximate maximum matching in 
the induced subgraph $G[U]$ (see~\Cref{prob:key}).  An additive $\eps n$ approximation to matching can also be turned into a multiplicative one, using 
another standard technique (see~\Cref{prop:additive}). Both these parts are by-now standard; see, e.g.~\cite{Kiss22,Behnezhad23,BhattacharyaKSW23,BhattacharyaKS23,AzarmehrBR24,Liu24,BehnezhadG24}.
The main part then is to solve~\Cref{prob:key}. 

\paragraph{The approach of~\cite{BehnezhadG24}.} The solution of \cite{BehnezhadG24} for~\Cref{prob:key} can be summarized as follows. The algorithm processes the updates in \emph{batches}. 
In each batch, $G_0$ will be the current graph at the start of the batch and 
the \emph{new} updates are inserted into a new graph $G_1$. There will also be a graph $G_2$ which is created by moving certain edges from $G_0$ (will be described later). 
 Given a query $U$, the algorithm tries to find a large matching in $(G_1 \cup G_2)[U]$ and if it succeeds, it returns that one and moves on. This step is done using the 
 greedy matching algorithm taking time linear in the size of $G_1 \cup G_2$. But, if $(G_1 \cup G_2)[U]$ does not have a large matching, the algorithm needs to search for a one in $G_0[U]$. This is done via a novel 
 \emph{sublinear time} ``opportunistic'' algorithm: the algorithm takes $O(n^{2+\eps}/d)$ time with high probability where $d$ is the \emph{average} degree of the graph $G_0[V(M)]$; i.e.,
 if the matching $M$ is ``far from'' being induced, then it can be found much faster than when it is close to being induced. The edges in this matching are then moved from $G_0$ to $G_2$. 
 
 The runtime analysis of this algorithm is as follows. Suppose each batch consists of $s$ updates. Since $O_{\eps}(1)$ queries are called every $\Theta_{\eps}(n)$ updates, the algorithm needs to handle $O_{\eps}(s/n)$ queries in a batch. 
 Moreover, both graphs $G_1$ and $G_2$ can only have $O_{\eps}(s)$ edges: the first one since there are $s$ updates in a batch and the second because each of the $O_{\eps}(s/n)$ queries may insert a matching inside $G_2$. 
 Thus, the entire time spent in this batch for running the greedy algorithm on $G_1 \cup G_2$ is $O_{\eps}(s^2/n)$ time. The runtime on $G_0$ however is calculated differently using a global argument. 
 Using the fact that $G_0$ only undergoes deletions during a batch,~\cite{BehnezhadG24} come up with an elegant analysis that shows that if the algorithm is spending a ``lot of time'' in finding ``near induced'' matchings in $G_0$, then, one can find a ``dense'' ORS graph in $G_0$ -- this allows for bounding the entire time the algorithm is spending on $G_0$ during this batch by $O(n^{2+\eps} \cdot \ORS{n}{\Theta_{\eps}(n)})$ time. This implies that by taking $s \approx n^{3/2} \cdot \ORS{n}{\Theta_{\eps}(n)}^{1/2}$, 
 the amortized update time of the algorithm will become $O_{\eps}(s/n)$ (over $s$ updates) which is $O_{\eps}(\sqrt{n^{1+\eps} \cdot \ORS{n}{\Theta_{\eps}(n)}})$ time. 

\paragraph{Our approach.} In the algorithm of~\cite{BehnezhadG24}, if the size of batch grows then so do the number of edges in the graph $G_1 \cup G_2$, and hence the cost of finding a greedy matching in $(G_1 \cup G_2)[U]$. However, the ORS-density based upper bound of the total running time of calls to the opportunistic algorithm remains unchanged with larger batches. Our main goal is to accelerate the step of finding matchings in $G_1 \cup G_2$ to allow for larger batches to amortize over the running cost of the opportunistic algorithm. To do this, observe that $G_1$ and $G_2$ are both dynamic graphs undergoing a small number of updates per each update to the underlying graph. This suggests that instead of statically finding matchings in $G_1 \cup G_2$ we may dynamically maintain them via an efficient dynamic matching algorithm.

We develop a \emph{recursive} variant of the algorithm of~\cite{BehnezhadG24}. The first main ingredient is a sublinear time algorithm that given a graph $G=(V,E)$ with $m$ edges and a query $U \subseteq V$, 
finds a large matching $M$ in $G$ in $O(m^{1+\eps}/d)$ time where $d$ is the \emph{maximum} degree of $M$ in $G[V(M)]$. Thus, this result strengthens the algorithm of~\cite{BehnezhadG24} in both relating it to the density of $G$ (instead of $O(n^2)$ always) 
and providing a stronger guarantee on the maximum internal degree of $M$ instead of the average degree.  The main step however is how to perform the recursion. 

We present a family of algorithms $\set{\alg_i}_{i \geq 1}$ with progressively better update times, where $\alg_1$ is similar to the algorithm of~\cite{BehnezhadG24} 
with a key difference of having $O_{\eps}(\sqrt{(m/n)^{1+\eps} \cdot \ORS{n}{\Theta_{\eps}(n)}})$ update time
instead (here, $m$ is a promised upper bound on the number of edges in $G$). 
Let us now consider constructing $\alg_2$ from $\alg_1$. 

We also process the inputs 
in batches of size $s$ with $G_0$ being the current graph at the start of the batch, $G_1$ receiving the updates during the batch, and $G_2$ receiving some removed matchings from $G_0$ while answering the queries. The main difference is that 
we are going to run $\alg_1$ on $G_1$ and $G_2$ separately instead of the greedy algorithm. These graphs now are going to have $O_{\eps}(s)$ edges in total (similar to what argued earlier) and thus the total runtime of processing
these graphs is $O_{\eps}(s \cdot \sqrt{(s/n)^{1+\eps} \cdot \ORS{n}{\Theta_{\eps}(n)}}$. If we find a large matching for a given query $U$ from either $G_1$ or $G_2$, we will be done, but if not, we need to rely on $G_0[U]$. 
In this case, we run our new sublinear time algorithm to find a matching with maximum internal degree $d$ in $O(m^{1+\eps}/d)$ time.  A similar argument as in~\cite{BehnezhadG24} (in fact considerably simpler given our stronger maximum degree guarantee) 
allows us to bound the total runtime of this step with $O_{\eps}(m^{1+\eps} \cdot \ORS{n}{\Theta_{\eps}(n)})$ time. Optimizing for the choice of $s$ then leads to an $O_{\eps}((m/n)^{1/3+\eps} \cdot \ORS{n}{\Theta_{\eps}(n)}^{2/3})$ 
amortized update time. 

Continuing like this for $\alg_3,\alg_4,\cdots$, while explicitly accounting for the loss in parameters (especially size of induced matchings in ORS graphs), gives us our final algorithm.



\section{Preliminaries}\label{sec:prelim}

\subsection{Basic Notation and Representation of Graphs}\label{sec:notation}
For a graph $G=(V,E)$, we use $n := \card{V}$ and $m:=\card{E}$. For a vertex $v \in V$, we use $N(v)$ to denote the neighbors of $v$ and $\deg(v) := \card{N(v)}$ to denote
its degree. For a subset $U \subseteq V$, $G[U]$ denotes the induced subgraph of $G$ on $U$. We use $\mu(G)$ to denote the maximum matching size in $G$. 

Since we will be designing sublinear-time algorithms for a dynamically changing graph, we briefly describe how the graphs will be represented to support various operations, namely, insertion and deletion of edges, neighbor queries, and pair queries, each in $O(1)$ expected time. For each vertex $u \in V$, we maintain a {\em dynamic} array $A_u$, a {\em dynamic} hash table $h_u$, and the current degree $\deg{(u)}$ of $u$. Here, dynamic refers to the property that at all times, the size of $A_u$ is $\Theta(\deg{(u)})$.

When an edge $(u,v)$ is inserted, we increment $\deg{(u)}$ by $1$, set $A_u[\deg{(u)}]=v$, and insert $v$ in the hash table $h_u$, storing along with it it the value $\deg{(u)}$, namely, the location of the vertex $v$ in the array $A_u$. This takes $O(1)$ expected time.

When an edge $(u,v)$ is deleted, let $w$ be the vertex at $A_u[\deg{(u)}]$. We decrement $\deg{(u)}$ by $1$. If $w = v$, we simply delete $v$ from the hash table $h_u$. Otherwise, let $j$ be the location of the vertex $v$ in $A_u$ which can be recovered using $h_u(v)$. We set $A_u[j] = w$, delete both $v$ and $w$ from $h_u$, and reinsert $w$ in $h_u$ associating $j$ to be its new location in $A_u$. This takes $O(1)$ expected time.

Finally, given any integer $i \in [\deg{(u)}]$, the task of outputting the $i$-{th} neighbor of $u$ is done by simply returning the vertex stored in $A_u[i]$. This also allows for sampling a random neighbor of $u$.

\subsection{Tools from Prior Work}\label{sec:tools} 

We start by recalling the following standard fact about the greedy algorithm for approximating maximum matchings. 

\begin{fact}\label{fact:greedy}
	Let $G=(V,E)$. The greedy algorithm that starts with $M = \emptyset$, iterates over edges of $G$ in any arbitrary order, and add an edge to $M$ if both its endpoints are currently unmatched, 
	returns a matching $M$ of size $\card{M} \geq 1/2 \cdot \mu(G)$ in $O(n+m)$ time. 
\end{fact}

\paragraph{Boosting frameworks for approximate matching.} 
We use standard boosting frameworks for obtaining a $(1-\eps)$-approximation to matching, using a ``weak'' approximation algorithm that only returns an $O(1)$-approximation. 
The original version of this framework is due to~\cite{McGregor05} and was de-randomized in~\cite{Tirodkar18}; for bipartite graphs, more efficient reductions are known in~\cite{AhnG11,AssadiLT21}. These results 
were tailored to additive approximation and dynamic graphs in~\cite{BhattacharyaKS23}. 

\begin{proposition}[\!\!\cite{McGregor05,AhnG11,Tirodkar18,AssadiLT21,BhattacharyaKS23}]\label{prop:boosting}
	Let $\gamma,\eps \in (0,1)$ be parameters. There exist functions $f(\gamma,\eps)$ and $g(\gamma,\eps)$ such that the following holds. 
	Let $\Algweak$ be an algorithm that given an $n$-vertex graph $G=(V,E)$ and any set $U \subseteq V$ of vertices with $\mu(G[U]) \geq f(\gamma,\eps) \cdot n$, returns a matching of size at least $\gamma \cdot f(\gamma,\eps) \cdot n$ in $G[U]$.  
	Then, there is a algorithm that given $G=(V,E)$ makes $g(\gamma,\eps)$ calls to $\Algweak$ on \underline{adaptively} chosen subsets of vertices, spending $O(f(\gamma,\eps) \cdot n)$ time preparing each subset, 
	and returns a matching of size $\mu(G) - \eps \cdot n$ in $G$. 
	
	For bipartite graphs, $f(\gamma,\eps) = \poly(\eps)$ and $g(\gamma,\eps) = \poly(1/(\gamma \cdot \eps))$ by~\cite{AhnG11,AssadiLT21}, while for general graphs, both $f(\gamma,\eps)^{-1}, g(\gamma,\eps) = (1/(\gamma \cdot \eps))^{O(1/(\gamma \cdot \eps))} = O_{\gamma,\eps}(1)$ by~\cite{McGregor05,Tirodkar18}. 
\end{proposition}

\paragraph{Sublinear-time estimation of maximum matching size.} We also rely on the sublinear-time algorithm of~\cite{Behnezhad21} for matching \emph{size} estimation. 

\begin{proposition}[\!\cite{Behnezhad21}]\label{prop:subtime-size}
	There is a randomized algorithm $\Algsub$ that for any $n$-vertex graph $G=(V,E)$ and a parameter $\eps \in (0,1)$, makes $\Ot(n \cdot \poly(1/\eps))$ queries to the adjacency matrix of $G$ and with high probability
	outputs an estimate $\tmu(G)$ such that 
	\[
		\frac{1}{2} \cdot \mu(G) - \eps \cdot n \leq \tmu(G) \leq \mu(G). 
	\]
\end{proposition}

\paragraph{Vertex sparsification.} Finally, we use vertex sparsification approaches of~\cite{AssadiKLY16,AssadiKL16,ChitnisCEHMMV16} as implemented by~\cite{Kiss22} for dynamic graphs. These sparsification approaches reduce 
the number of vertices to $O(\mu(G)/\eps)$ via vertex contraction while preserving a $(1-\eps)$-approximate maximum matching in the graph. This allows one to turn additive approximation to matching into a multiplicative one, with minor overhead (see, e.g.~\cite{BehnezhadG24}, for an example of how this is used).  

\begin{proposition}[\!\cite{AssadiKLY16,AssadiKL16,ChitnisCEHMMV16,Kiss22}]\label{prop:additive}
	Suppose $\Algadd$ is an algorithm that given a parameter $\eps > 0$ can process a fully dynamic $n$-vertex graph $G=(V,E)$ and maintains a matching of size at least $\mu(G) - \eps \cdot n$ in $T(n,\eps)$ amortized update time. 
	Then, there is a randomized algorithm that can with high probability maintain a $(1-\eps)$-approximation to maximum matching in $n$-vertex fully dynamic graphs in $O(T(n,\Theta(\eps^2)) \cdot \poly(\log{(n)}/\eps))$ amortized update time. 
\end{proposition}

\subsection{An Auxiliary Lemma on ORS Graphs}\label{sec:ors-lemma}

For a matching $M$ in a graph $G$, we define the \textbf{maximum internal degree} $\indeg{M}$ of $M$, as the maximum degree of the induced subgraph $G[V(M)]$. 
This way, $\indeg{M} = 1$ iff $M$ is an induced matching, and in general, smaller the value $\indeg{M}$, the ``closer'' $M$ is to an induced matching. 

The following lemma is a simplification of a similar result in~\cite[Lemma 14]{BehnezhadG24} (which even works for average internal degree instead of maximum degree), using a simple variant of the rounding approach of~\cite{GoelKK12} for RS graphs. 

\begin{lemma}\label{lem:ORS}
	Let $G=(V,E)$ be any graph and $\MM := (M_1,\ldots,M_\rho)$ be an ordered set of matchings in $G$, each of size $\ell$, such that for every $i \in [\rho]$, $\indeg{M_i}$ in the subgraph of $G$ on edges $M_i \cup M_{i+1} \cup \ldots \cup M_{\rho}$ is some given $d_i \geq 1$. 
	Then, for every $\eta \in (0,1/100)$, 
	\[
		\sum_{i=1}^{\rho} \frac{1}{d_i} \leq \frac{34\log{n}}{\eta} \cdot \ORS{n}{(1-\eta) \cdot \ell}. 
	\]
\end{lemma}
\begin{proof}
	For any integer $d \in \set{1,2,4,\ldots,n}$, define 
	$
		\MM(d) := \set{M_i \in \MM \mid d \leq d_i < 2d_i}.
	$
	Moreover, let 
	\[
	\dstar := \argmax_{d} \sum_{M_i \in \MM(d)} {1}/{d_i}.
	\]
	Since $\MM(1),\MM(2),\cdots$ partition $\MM$, we have that 
	\begin{align}
		\card{\MM(\dstar)} = \dstar \cdot \sum_{M_i \in \MM(\dstar)} \frac{1}{\dstar} \geq \dstar \cdot \sum_{M_i \in \MM(d)} \frac{1}{d_i} \geq \frac{\dstar}{\log{n}} \cdot \sum_{i=1}^{\rho} \frac{1}{d_i}. \label{eq:mm-dstar}
	\end{align}
	
	We now show that a random subset of $\MM(\dstar)$ indeed forms an ORS graph with induced matchings of size $(1-\eta) \cdot \ell$ which will be enough to conclude the proof. 
	
	Pick $\MMstar = (M_1,\ldots,M_{\rho'}) \subseteq \MM(\dstar)$  wherein each matching of $\MM(\dstar)$ is chosen with probability 
	\[
	p := \frac{\eta}{20\dstar}
	\]
	independently (with the sampled matchings ordered in the same manner as they were in $\MM$). 
	Fix any matching $M_i$ in $\MMstar$. For any vertex $v$ matched by $M_i$, define $\degstar_i(v)$ as the degree of $v$ in the subgraph $M_{i+1},\ldots,M_{\rho'} \in \MMstar$ (notice that this \emph{excludes} degree of $v$ in $M_i$ itself).  
	We have, 
	\[
		\expect{\degstar_i(v)} = \sum_{\substack{j > i \in \MM(\dstar) \\ \text{$v$ has an edge in $M_j$}}} \Pr\paren{\text{$M_j$ is chosen in $\MMstar$}} \leq 2\dstar \cdot \frac{\eta}{20\dstar} = \frac{\eta}{10},
	\]
	where in the inequality uses the fact that degree of all vertices in $\MM(\dstar)$ in the suffix matchings is at most $2\dstar$ in the entire $\MM$ (and thus among $\MM(\dstar)$ for sure). 
	
	We now say that $v$ is \textbf{bad} for $M_i$ iff $\degstar_i(v) \geq 1$. By the above calculation
	and a Markov bound,  the probability that $v$ is bad for $M_i$ is at most $\eta/10$. 
	This means that the expected number of bad vertices for $M_i$ is at most $2\ell \cdot \eta/10 = \ell \cdot \eta/5$. 
	
	We further say that the matching $M_i$ itself is \textbf{bad} if it has 
	at least $\eta \cdot \ell$ bad vertices. Another application of Markov bound implies that the probability $M_i$ is bad is at most $1/5$. 
	This means that in expectation, at least $4/5$ of the sampled matchings in $\MMstar$ are not bad. 
	
	By the probabilistic method (and since the size of $\MMstar$ is concentrated), 
	there exist a set of $3/5 \cdot p \cdot \card{\MM(\dstar)}$ matchings in $\MMstar$ none of which are bad. For each of these matchings, remove all their bad vertices which reduces 
	their size to $(1-\eta) \cdot \ell$ in the worst case and arbitrarily remove more edges such that all of them have size $(1-\eta) \cdot \ell$ exactly. By definition, the resulting graph is now an $(r,t)$-ORS graph 
	with parameters 
	\[
		r = (1-\eta) \cdot \ell, \quad \text{and} \quad  t = \frac{3}{5} \cdot \frac{\eta}{20 \dstar} \cdot \card{\MM(\dstar)}. 
	\]
	By definition, we have $t \leq \ORS{n}{(1-\eta) \cdot \ell}$ which implies that 
	\[
		\card{\MM(\dstar)} \leq \frac{34 \dstar}{\eta} \cdot \ORS{n}{(1-\eta) \cdot \ell}. 
	\]
	Plugging in this bound in~\Cref{eq:mm-dstar} concludes the proof. 
\end{proof}


\section{An Opportunistic Sublinear-Time Algorithm for Matching}\label{sec:subtime} 

In this section, we provide one of the key subroutines used by our dynamic algorithm. This subroutine takes as input a ``base'' \emph{static} graph $G$ and a set $U$ with the promise that $G[U]$ is of size $\Omega(n)$, and outputs a matching of size $\Omega(n)$ in $G[U]$. While the worst-case runtime of this algorithm is (almost) linear in the number of edges of the base graph, it can be much better if the maximum internal degree of the matching it outputs is large. 

\begin{lemma}\label{thm:subtime}
	There is an algorithm (\Cref{alg:subtime}) that given an $n$-vertex graph $G=(V,E)$ with $m$ edges, parameters $\gamma,\delta \in (0,1/6)$, and vertices $U \subseteq V$ with $\mu(G[U]) \geq \delta n$, 
	with high probability returns a matching $M$ in $G[U]$ with size at least $\gamma \cdot \delta n$ in $O(m \cdot n^{3\gamma} \cdot \log{(n)}/\indeg{M})$ time. 
\end{lemma}

We note that the \emph{worst-case} bound of $\Omega(m)$ in~\Cref{thm:subtime}
is necessary due to a lower bound of~\cite{AssadiCK19} (as opposed to~\Cref{prop:subtime-size} for \emph{size} estimation); however, in the hard instances of that lower bound, one necessarily needs to find a large matching with maximum internal degree $O(1)$, which means, one cannot benefit from the extra power of this lemma (as expected).

This result is inspired by~\cite[Lemma 9]{BehnezhadG24} and generalizes and strengthens it. Algorithm of~\cite{BehnezhadG24} finds a matching of size $\gamma \cdot \delta n$ in $O(n^{2+O(\gamma)}/d)$ time (thus, does not benefit from the number of edges of $G$) and moreover $d$ is the \emph{average} degree of $G[V(M)]$ instead of our stronger guarantee on maximum degree. To obtain our result, we use an argument similar in spirit 
to that of ``residual sparsity guarantee'' of the greedy algorithm used for maximal independent set and maximal matching problems, e.g., in~\cite{AhnCGMW15,Konrad18,AssadiOSS19}.

The algorithm in~\Cref{thm:subtime} works by sampling the edges of $G$ with geometrically increasing probabilities and only consider edges 
of $G$ that are sampled and belong to $G[U]$. It then attempts to find a ``large'' matching in this sampled graph using the greedy algorithm; if it succeeds, it returns \emph{this} matching only, otherwise, it will remove vertices of this matching and continues 
to the next sampling phase. Formally, the algorithm is as follows. 

\begin{Algorithm}[The algorithm of~\Cref{thm:subtime}]\label{alg:subtime}
~  

\noindent
{\bf Input:} A graph $G(V,E)$, a set $U \subseteq V$, s.t. $\mu(G[U]) \ge \delta n$, and $\gamma,\delta \in (0,1/6)$.\\
{\bf Output:} A matching $M$ in $G[U]$ of size at least $\gamma \cdot \delta n$.

	\begin{itemize}[leftmargin=10pt]
		\item Let $X_1 = U$. For $i=1$ to $\Gamma := 1/(3\gamma)$ iterations: 
		\begin{enumerate}
			\item Let the sampling probability be $p_i := n^{(3\gamma) \cdot i}/n$. For every vertex $v \in X_i$, pick a set $N_i(v) \subseteq N(v)$ by sampling each neighbor of $v$ in $G$ independently 
			 with probability $p_i$. 
			\item Start with $M_i = \emptyset$. Iterate over vertices of $X_i$ in a fixed order, and for each $v \in X_i$, if both $v$ and some $w$ in $N_i(v) \cap X_i$ are unmatched by $M_i$, then add the edge $(v, w)$ to $M_i$.  
			\item If $\card{M_i} \geq \gamma \cdot \delta n$, return $M_i$ and {terminate}. Otherwise, let $X_{i+1} = X_i \setminus V(M_i)$, and continue to the next sampling step. 
		\end{enumerate}
	\end{itemize}
\end{Algorithm}

We start by arguing that the algorithm always returns a matching in one of its iterations. 

\begin{claim}\label{clm:subtime-terminate}
	\Cref{alg:subtime} always returns a matching $M_i$ of size at least $\gamma \cdot \delta n$ in some  $i \in  [\Gamma]$. 
\end{claim}
\begin{proof}
	Suppose the algorithm has not terminated until the beginning of the last iteration. This means that for all $i < \Gamma$, we have, $\card{M_i} < \gamma \cdot \delta n$. Thus, the total number 
	of vertices from $U$ that are removed until reaching $X_{\Gamma}$ is 
	\[
		\sum_{i=1}^{\Gamma-1} \card{V(M_i)} <  \Gamma \cdot 2 \cdot (\gamma \cdot \delta n) = 2\delta n/3,
	\]
	given $\Gamma = 1/(3\gamma)$. 
	Given that $G[U]$ has a matching of size at least $\delta n$ by the theorem statement, there is still a matching of size at least $\delta n - 2\delta n/3 = \delta n/3$ inside $G[X_{\Gamma}]$ (after removing the vertices counted above). 
	But, in iteration $i = \Gamma$, every edge is sampled with probability $n^{3\gamma \cdot 1/(3\gamma)}/n = 1$ and thus $N_i(v)$ is the entire neighborhood of $v$ in $G[X_{\Gamma}]$.  Hence, the greedy algorithm necessarily finds a matching of size at least $\delta n/6 > \gamma \cdot \delta n$ in this case (since $\gamma < 1/6$). Thus the algorithm terminates in this last iteration and outputs $M_{\Gamma}$ as the answer.  
\end{proof}

The next step is to bound the runtime of the algorithm based on the iteration it terminates in. 

\begin{claim}\label{clm:subtime-runtime}
	With high probability, for every $i \in [\Gamma]$, if~\Cref{alg:subtime} returns the matching $M_i$ and terminates in this iteration, then its runtime is $O(m \cdot p_i)$. 
\end{claim}
\begin{proof}
	Firstly, with high probability, the runtime of the algorithm in each iteration $j \in [\Gamma]$ is $O(m \cdot p_j)$. This is because with high probability, the number of edges in $G_j$ is $O(m \cdot p_j)$ by a simple application of Chernoff bound, 
	and the edges can be sampled in this much time using standard ideas instead of explicitly going over each edge and sampling them\footnote{For each vertex $v$, first sample a number $k_v$ from the binomial distribution of $\deg(v)$ and $p_i$ (using its closed-form formula); then, sample $k_v$ neighbors of $v$ uniformly at random from $N(v)$.}. Running the greedy matching algorithm also take another $O(m \cdot p_j)$ time now. 
	Finally, since $p_j$'s form a geometric series (as $p_{j+1} = n^{3\gamma} \cdot p_j$), we have $\sum_{j=1}^{i} p_j = O(p_i)$ which concludes the proof. 
\end{proof}

We now bound $\indeg{M_i}$ for the matching $M_i$ returned by the algorithm. Instead of bounding the maximum internal degree of the matching itself, we simply bound the maximum degree of the subgraph $G[X_i]$ where the matching $M_i$ is chosen from. 

\begin{claim}\label{clm:subtime-density}
	With high probability, for every $i \in [\Gamma]$, max-degree of $G[X_i]$ is $O(n^{3\gamma} \cdot \log{(n)}/p_i)$. 
\end{claim}
\begin{proof}
	The claim trivially holds for $i=1$ since $p_1 = n^{3\gamma}/n$ and $G[X_1]$ can only have maximum degree $n < n\log{(n)} = n^{3\gamma}\log{(n)}/p_1$.  
	We focus on the $i > 1$ case in the following. 
	
	Fix any vertex $v$ in $G[X_{i-1}]$. Consider the step wherein $v$ is being processed by the greedy algorithm. 
	First, suppose that the degree of $v$ to vertices in $X_{i-1} \setminus V(M_{i-1})$ at this point is at most $100\ln{(n)}/p_{i-1}$. In this case, 
	even if $v$ remains unmatched,  its degree in $G[X_i]$ will be $O(\log{(n)}/p_{i-1}) = O(n^{3\gamma} \cdot \log{(n)}/p_i)$ as desired. 
	
	On the other hand, consider the case where degree of $v$ to vertices of $X_{i-1} \setminus V(M_{i-1})$ is at least $100\ln{(n)}/p_{i-1}$ when we start processing $v$. 
	Then, the probability that none of these neighbors are sampled in $N_i(v)$ is at most 
	\[
		(1-p_{i-1})^{100\log{(n)}/p_{i-1}} \leq e^{-100\ln{n}} = n^{-100}; 
	\]
	thus, with high probability, at least one of these vertices is sampled in $N_i(v)$. Conditioned on this event, $v$ will surely get matched (if it was not already matched) by the greedy algorithm. 
	Taking a union bound over all vertices now ensures that with high probability, every vertex in $X_{i-1}$ that remains unmatched by $M_{i-1}$ has degree $O(n^{3\gamma} \cdot \log{(n)}/p_i)$ to other unmatched 
	vertices in $X_{i-1}$, and hence in the graph $G[X_i]$. This concludes the proof. 
\end{proof}

We are now ready to conclude the proof of~\Cref{thm:subtime}. 

\begin{proof}[Proof of~\Cref{thm:subtime}]
	By~\Cref{clm:subtime-terminate}, we know~\Cref{alg:subtime} terminates in some iteration $i \in [\Gamma]$ and returns a matching $M_i$ of size at least $\gamma \cdot \delta n$ by the termination condition. 
	By~\Cref{clm:subtime-runtime}, this takes $O(m \cdot p_i)$ time with high probability. 
	
	Finally, since $V(M_i)$ is a subset of $X_i$, we obtain that $\indeg{M_i}$ is at most the maximum degree of $G[X_i]$ 
	which itself is at most $O(n^{3\gamma} \cdot \log{(n)}/p_i)$ with high probability by~\Cref{clm:subtime-density}. Thus, the runtime of the algorithm is $O(m \cdot n^{3\gamma} \cdot \log{(n)}/\indeg{M_i})$ with high probability as desired. 
\end{proof}



\clearpage

\newcommand{\Gold}{\ensuremath{G_{\textnormal{old}}}}

\newcommand{\Gbatch}{\ensuremath{G_{\textnormal{batch}}}}

\newcommand{\Gmatch}{\ensuremath{G_{\textnormal{match}}}}

\newcommand{\Mold}{\ensuremath{M_{\textnormal{old}}}}

\newcommand{\Mbatch}{\ensuremath{M_{\textnormal{batch}}}}

\newcommand{\Mmatch}{\ensuremath{M_{\textnormal{match}}}}

\section{A Key Intermediate Dynamic Problem}\label{sec:intermediate}

In order to prove our main result, we focus on solving the following intermediate problem. Similar versions of this problem also appear in recent work including~\cite{BhattacharyaKS23,Liu24,BehnezhadG24}. 

\begin{Problem}\label{prob:key}
	The problem is parameterized by integers $n,m,q \geq 1$ and reals $\gamma,\delta,\alpha \in (0,1)$. We have a fully dynamic $n$-vertex graph $G=(V,E)$ that starts empty, i.e., has $E = \emptyset$, and throughout, never has more than $m$ edges,
	nor receives more than $\poly(n)$ updates in total. 
	
	\medskip
	
	\emph{\textbf{Updates:}}  The updates to $G$ happen in \emph{\textbf{chunks}} $C_1,C_2,\ldots$, each consisting of exactly $\alpha \cdot n$ edge insertions or deletions in $G$. 
	
	\emph{\textbf{Queries:}} After each chunk, there will be at most $q$ queries, coming one at a time and in an adaptive manner (based on the answer to all prior 
	queries including the ones in this chunk).  Each query is a set $U \subseteq V$ of vertices with the promise that $\mu(G[U]) \geq \delta n$; the algorithm should respond
	with a matching of size at least $\gamma \cdot \delta n$ from $\mu(G[U])$. 
	
	For ease of reference, we list the parameters of this problem and their definitions: 
	\begin{align*}
		n &: \text{number of vertices in the graph}; \\
		m &: \text{maximum number of edges at any point present in the graph}; \\
		q &: \text{number of adaptive queries made after each chunk}; \\
		\gamma &: \text{approximation ratio of the returned matching for each query}; \\
		\delta &: \text{a lower bound on the fraction of vertices matched in the subgraph of $G$ for the query}; \\
		\alpha &: \text{a parameter for determining the size of each chunk as a function of $n$}.
	\end{align*}
	
	\noindent
	For technical reasons, we allow additional updates, called \emph{\textbf{empty updates}} to also appear in the chunks but these ``updates'' do not change any edge of the graph, although will be 
	counted toward the number of updates in their chunks\footnote{This is used for simplifying the exposition when solving this problem recursively; these empty updates will still be counted when computing amortized runtime of these 
	recursive algorithms.}. 
\end{Problem}

We will design a family of recursive algorithms for solving~\Cref{prob:key} in this section, starting with the base case, which also acts as a good warm-up for the key ideas of the algorithm. 

\subsection{Base Case}\label{sec:base}

The proof of the following lemma follows a similar approach as used in~\cite{BehnezhadG24} albeit with several modifications to take into account the dependence on the (overall) sparsity of the input graph and to match
the requirements of~\Cref{prob:key}. This lemma effectively gives an algorithm with $\approx \sqrt{(m/n \cdot \ORS{n}{\Theta(n)})}$ update time for solving~\Cref{prob:key}. 

\begin{lemma}\label{lem:base-case}
	There is an algorithm (\Cref{alg:base-case}) for~\Cref{prob:key} that with high probability takes
	\[
		O\Paren{q \cdot \sqrt{\frac{m \cdot n^{6\gamma} \cdot \ORS{n}{\gamma \cdot \delta n/2}}{\alpha \cdot n}}},
	\]
	amortized time over the updates to maintain the answer to all given queries in each chunk. The algorithm works as long as $\gamma < 1/12$ and $\alpha \geq \gamma \cdot \delta$. 
\end{lemma}

The algorithm in~\Cref{lem:base-case} processes the chunks in \textbf{batches}. Each batch $B$ processes $t$ chunks of updates to $G$ for some $t$ to be fixed later ($t$ is going to be $\approx \sqrt{(m/n \cdot \ORS{n}{\Theta(n)})}$). Whenever a batch starts, the algorithm maintains threes graphs (see also~\Cref{fig:schematic1} for an illustration): 
\begin{itemize}
\item $\Gold$ which starts as the graph $G$ at the beginning of the batch; no further insertions will be added to $\Gold$ during the processing of this batch and if an edge already in $\Gold$ is deleted in an update, the algorithm removes the edge from $\Gold$ (i.e., $\Gold$ 
is a decremental graph); 
\item $\Gbatch$ which starts as an empty graph and receives all subsequent insertion of edges to $G$ during the updates of this batch and will be updated based on their deletions also;
\item $\Gmatch$ which starts as an empty graph and is updated by the algorithm by moving certain matchings from $\Gold$ to $\Gmatch$ instead. Once an edge is moved to $\Gmatch$, if it gets deleted, 
the algorithm deletes the edge from $\Gmatch$ (and subsequent insertions are processed in $\Gbatch$; in other words, insertions to $\Gmatch$ only come from moving edges from $\Gold$ to $\Gmatch$). 
\end{itemize}

Given a query $U \subseteq V$, the algorithm starts by examining
 the edges of $\Gbatch$ to see if it can already find a large matching in $\Gbatch[U]$. This is done by running the greedy matching algorithm (in~\Cref{fact:greedy}) over the edges of $\Gbatch[U]$ to obtain a matching $\Mbatch$. 
 If $\Mbatch$ is large enough, it will be returned as the answer to the query. Otherwise, the algorithm runs the greedy matching algorithm on $\Gmatch[U]$ to obtain a matching $\Mmatch$. Again, if this matching is large enough, 
 it will be returned as the answer to the query. Finally, if neither of these cases happen, then the algorithm runs~\Cref{alg:subtime} on $\Gold$ with the subgraph $U$ to obtain a matching $\Mold$; 
 we can guarantee this matching is large enough given $G[U]$ is promised to have a large matching. The algorithm then moves all edges of $\Mold$ from $\Gold$ to $\Gmatch$ and returns $\Mold$ as the answer to the query.
 
 A formal specification of the algorithm is as follows. 
 
 \begin{Algorithm}[Algorithm of~\Cref{lem:base-case}]\label{alg:base-case}
 	~
	\vspace{-5pt}
	\begin{itemize}
		\item Process the updates in batches $B$ of $t$ chunks $C_1,\ldots,C_t$ and for each batch: 
		\begin{enumerate}
			\item Let $\Gold = G$, $\Gbatch = \emptyset$, and $\Gmatch = \emptyset$ (on vertices $V$). Maintain these graphs as follows for updates in each chunk (these graphs may also be updated based on queries): 
			\begin{enumerate}
			\item For an edge insertion $e=(u,v)$, add the edge to $\Gbatch$. 
			\item For an edge deletion $e=(u,v)$, remove the edge $e$ from each of the graphs $\Gold$, $\Gbatch$, or $\Gmatch$ that it belongs to currently. 
			\end{enumerate}
			\item After a chunk is updated, answer each query $U \subseteq V$ as follows: 
			\begin{enumerate}
				\item\label{line:base-batch} Go over all edges of $\Gbatch$ and run the greedy matching algorithm on $\Gbatch[U]$ to obtain a matching $\Mbatch$. If $\card{\Mbatch} \geq \gamma \cdot \delta  n$, return $\Mbatch$, otherwise continue. 
				\item\label{line:base-match} Go over all edges of $\Gmatch$ and run the greedy matching algorithm on $\Gmatch[U]$ to obtain a matching $\Mmatch$. If $\card{\Mmatch} \geq \gamma \cdot \delta  n$, return $\Mmatch$, otherwise continue. 
				\item\label{line:base-old} Run~\Cref{alg:subtime} on $\Gold$, the set $U$, and parameter $2\gamma$ with the guarantee that $\mu(\Gold[U]) \geq \delta n/2$ (which we establish in~\Cref{clm:base-case-correct}) 
				to obtain a matching $\Mold$. Move $\Mold$ from $\Gold$ to $\Gmatch$ and return $\Mold$. 
			\end{enumerate}
		\end{enumerate}
	\end{itemize}
 \end{Algorithm}
 
 \begin{figure}[t!]
 \centering
 \includegraphics[scale=0.4]{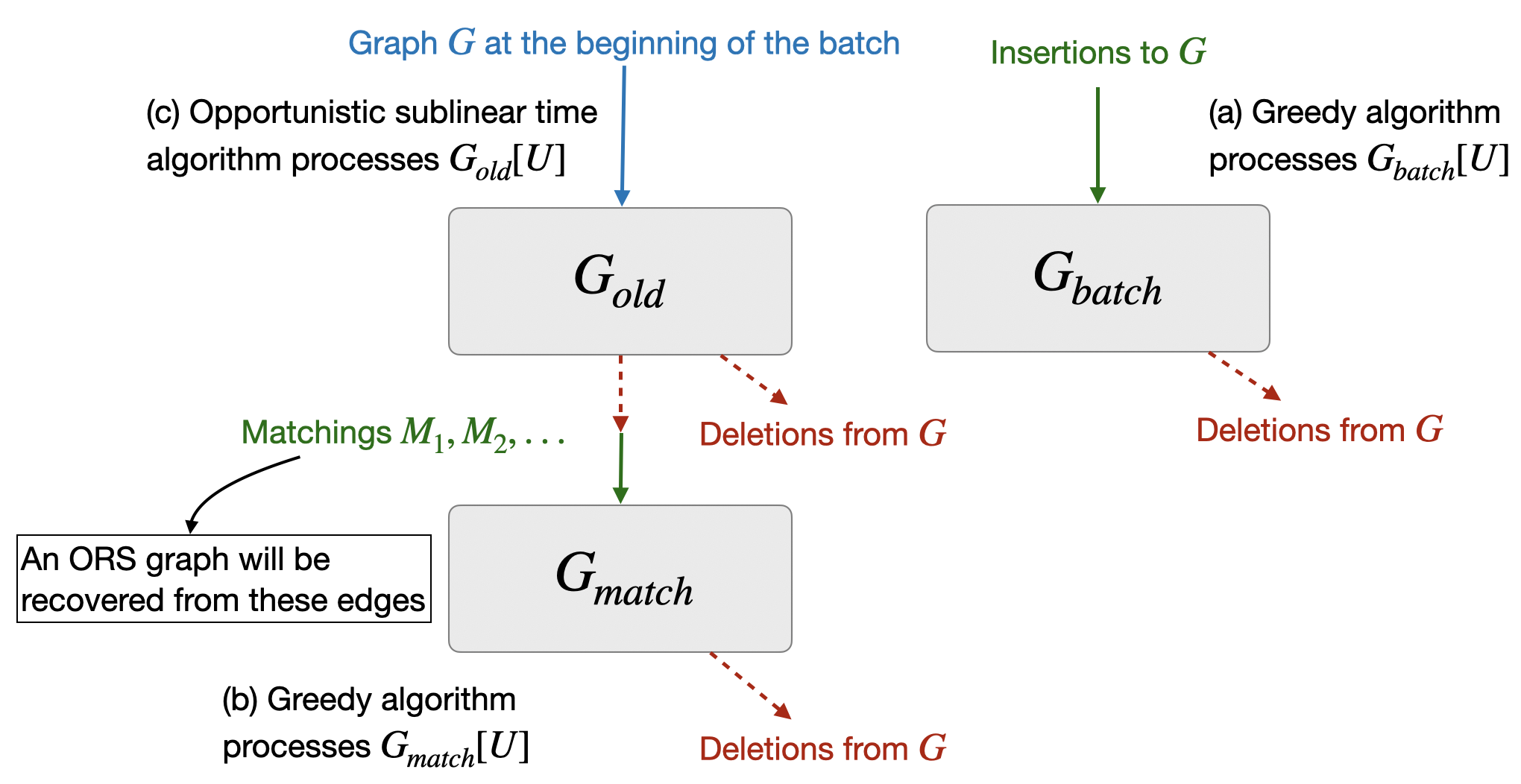}
 \caption{An illustration of the three graphs $\Gold,\Gbatch,\Gmatch$ in~\Cref{alg:base-case}, their role, and how they are being processed.
 Notice that $\Gold$ is a decremental graph, while $\Gbatch,\Gmatch$ are fully dynamic. The analysis of the algorithm forms an ORS from
 the edges of the matchings $M_1,M_2,\ldots,$ moved from $\Gold$ to $\Gmatch$ -- this ORS is a \emph{subgraph} of the \emph{static} graph $G$ at the beginning of the batch
 and does not contain any edges inserted in this batch.
 }\label{fig:schematic1}
 \end{figure}
 
 We start by establishing the correctness of~\Cref{alg:base-case}. 
 
 \begin{claim}\label{clm:base-case-correct}
 	With high probability, the answer to each query $U$ in \Cref{alg:base-case} is a valid answer according to~\Cref{prob:key}. 
 \end{claim}
 \begin{proof}
 	Notice that $\Gold,\Gbatch,\Gmatch$ at any point partition the current graph $G$. If either of $\Mbatch$ or $\Mmatch$ is of size at least $\gamma \cdot \delta n$, the output is correct. 
	Otherwise, by the guarantee of the greedy algorithm, we know that both $\mu(\Gbatch[U]), \mu(\Gmatch[U])$ are at most $2\gamma \cdot \delta n$. Thus,  
	\[
		\mu(\Gold[U]) \geq \mu(G[U]) - \mu(\Gbatch[U]) - \mu(\Gmatch[U]) \geq \delta n - 4 \gamma \cdot \delta n \geq \delta n/2,
	\]
	by the choice of $\gamma < 1/12$. This implies that the requirement of~\Cref{thm:subtime} for $\Gold[U]$ is satisfied (including having $2\gamma < 1/6$) and thus, its output, with high probability, is of size $2 \gamma \cdot \delta n/2 = \gamma \cdot \delta n$. 
	Thus, the returned matching in this case is also of the proper size, concluding the proof. 
 \end{proof}
 
The main part is to analyze the running time of~\Cref{alg:base-case}. The following claim is the key to relating the runtime of this algorithm to the density of ORS graphs. 

\begin{claim}\label{clm:base-case-ORS}
	Let $M_1,M_2,\ldots,M_{\rho}$ be the matchings computed from $\Gold$ in Line~\eqref{line:base-old} of~\Cref{alg:base-case} and added to $\Gmatch$ \underline{at the time of their computation} (i.e., here, we ignore the deletions that have happened subsequently, namely, some edges of $M_i$ might have been deleted from $\Gold$ when we are inserting $M_{i+1}$, but we still keep those edges in the definition of $M_i$).  
	These matchings are edge-disjoint and for every $i \in [\rho]$, maximum degree of $M_i$ among the matchings $M_{i},\ldots,M_{\rho}$ is at most $\indeg{M_i}$ in the graph $\Gold$ at the time $M_i$ was computed. 
\end{claim}
\begin{proof}
	The edge-disjoint part follows from the fact that each $M_i$ is chosen from $\Gold$ and at that point, edges of $M_1,\ldots,M_{i-1}$ are already removed from $\Gold$ (and cannot be inserted back to $\Gold$). 
	
	Fix any $i \in [\rho]$ and matching $M_i$. Since $\Gold$ is a decremental graph, all edges in $M_{i+1},\ldots,M_{\rho}$ belong to $\Gold$ at the time of computation of $M_i$. Thus, these edges will be counted toward 
	$\indeg{M_i}$ in $\Gold$ at the time of computation of $M_i$. As such, the maximum degree of $M_i$ among the matchings $M_{i},\ldots,M_{\rho}$ is at most $\indeg{M_i}$ as desired. 
\end{proof}

We can now bound the runtime of the algorithm. 

\begin{claim}\label{clm:base-case-runtime}
	With high probability, when running~\Cref{alg:base-case} on a single batch of $t$ chunks: 
	\begin{enumerate}
		\item the total time spent for maintaining the graphs and bookkeeping is $O(t \cdot q \cdot \alpha \cdot n)$ time; 
		\item the total time spent computing $\Mbatch$ in Line~\eqref{line:base-batch} is $O(t \cdot q \cdot t \cdot \alpha \cdot n)$ time; 
		\item the total time spent computing $\Mmatch$ in Line~\eqref{line:base-match} is $O(t \cdot q \cdot t \cdot q \cdot \alpha \cdot n)$ time; 
		\item the total time spent computing $\Mold$ in Line~\eqref{line:base-old} is $O(m \cdot n^{6\gamma} \cdot \log^2\!{(n)} \cdot \ORS{n}{\gamma \cdot \delta n/2})$ time. 
	\end{enumerate}
\end{claim}
\begin{proof}
	There are $t$ chunks in a batch, each involving $\alpha  n$ updates and $q$ queries. Processing the updates can be done in $O(1)$ expected time for each update by maintaining the three graphs $\Gold, \Gbatch, \Gmatch$ 
	as explained in~\Cref{sec:notation}. 
	Thus, with high probability, these steps take $O(t \cdot q \cdot \alpha \cdot n)$ time in total. Finally, given $\alpha \geq \gamma \cdot \delta$ in~\Cref{lem:base-case}, 
	moving each choice of $\Mold$ from $\Gold$ to $\Gmatch$ takes $O(\alpha \cdot n)$ time per each query (at most), and thus $O(t \cdot q \cdot \alpha \cdot n)$ in total.
	
	For computing $\Mbatch$ for each query, the algorithm iterates over edges of $\Gbatch$ and takes linear time in the size of the entire $\Gbatch$ to run the greedy matching algorithm (on $\Gbatch[U]$). 
	Given that $\Gbatch$ can only have $\leq t \cdot \alpha \cdot n$ edges at any point (before a new batch is restarted), 
	the runtime for each query is $O(t \cdot \alpha \cdot n)$ time. Given there are $t \cdot q$ queries in total, this part takes  $O(t \cdot q \cdot t \cdot \alpha \cdot n)$ time. 
	
	Similarly, computing $\Mmatch$ for each query is done by iterating over all edges of $\Gmatch$ and taking linear time on those. The edges in $\Gmatch$ come from inserting a matching of size $\gamma \cdot \delta n \leq \alpha n$ (by the assumption
	in the statement of~\Cref{lem:base-case}) after a query (possibly) and since there are most $t \cdot q$ queries, 
	there can be at most $O(t \cdot q \cdot \alpha \cdot n)$ edges in $\Gmatch$. Thus, similar (but not identical) to the previous case, this step takes $O(t \cdot q \cdot t \cdot q \cdot \alpha \cdot n)$ time (this is a factor $q$ larger). 
	
	We now get to the main part of bounding the runtime of computing $\Mold$. Let $M_1,M_2,\ldots,M_{\rho}$ be the matchings computed as $\Mold$ throughout this entire batch. For each $i \in [\rho]$, let $\indeg{M_i}$
	denote the maximum degree of $M_i$ in $\Gold$ at the time $M_i$ was computed; additionally, let $d_i$ denote the maximum degree of $M_i$ among the matchings $M_i, \ldots, M_{\rho}$. 
	By~\Cref{clm:base-case-ORS}, we have $\indeg{M_i} \geq d_i$. Moreover, by~\Cref{thm:subtime}, the runtime for computing $M_i$ in~\Cref{alg:subtime} with high probability is $O(m \cdot n^{6\gamma} \cdot \log{(n)}/\indeg{M_i})$. 
	Putting these together with~\Cref{clm:base-case-ORS} in~\Cref{lem:ORS} (for parameter $\eta=1/2$, and since size of each matching is $\gamma \cdot \delta n$ by~\Cref{clm:base-case-ORS}), we have that the total time spent computing $M_1,\ldots,M_{\rho}$ is
	\begin{align*}
		\sum_{i=1}^{\rho} O(m \cdot n^{6\gamma} \cdot \log{(n)} \cdot \frac{1}{\indeg{M_i}}) &\leq O(m \cdot n^{6\gamma} \cdot \log\!{(n)}) \cdot \sum_{i=1}^{\rho} \frac{1}{d_i} \\
		&\leq O(m \cdot n^{6\gamma} \cdot \log^2\!{(n)} \cdot \ORS{n}{\gamma \cdot \delta n/2}). 
	\end{align*}
	This concludes the proof. 
\end{proof}


\begin{proof}[Proof of~\Cref{lem:base-case}]
	The correctness of the algorithm follows from~\Cref{clm:base-case-correct} and a union bound over $\poly(n)$ intermediate graphs created in~\Cref{prob:key} (by the assumption on number of updates).  
	
	Furthermore, the amortized runtime per each of $t \cdot \alpha \cdot n$ updates during a batch, by~\Cref{clm:base-case-runtime} is 
	\[
		O(t \cdot q^2) + O(\frac{1}{t \cdot \alpha \cdot n} \cdot m \cdot n^{6\gamma} \cdot \log^2\!{(n)} \cdot \ORS{n}{\gamma \cdot \delta n/2}). 
	\]
	We can now balance these terms by setting 
	\[
		t := \paren{\frac{m \cdot n^{6\gamma} \cdot \ORS{n}{\gamma \cdot \delta n/2}}{\alpha \cdot n \cdot q^2}}^{1/2},
	\]
	which leads to the desired update time of 
	\[
		O(q) \cdot \paren{\frac{m \cdot n^{6\gamma} \cdot \ORS{n}{\gamma \cdot \delta n/2}}{\alpha \cdot n}}^{1/2}.
	\]
	Note however that it is possible the \emph{entire} number of updates to~\Cref{prob:key} is less than $t$, namely, the algorithm does not receive even one full batch of updates. 
	In that case, we cannot amortize the runtime as above given the fewer number of updates.
	
	Nevertheless, since by the definition of~\Cref{prob:key}, the graph starts as an empty graph, in this case both $\Gold = \emptyset$ and $\Gmatch = \emptyset$ for the single batch processed by the algorithm. Thus, the entire runtime of the algorithm will be based on the first two items of~\Cref{clm:base-case-runtime} and thus is still upper bounded as above. 
\end{proof}

\subsection{The Recursive Step} \label{sec:recursive} 

We now design a family of recursive algorithms $\set{\alg}_{k=1}^{\infty}$ with progressively better update times for~\Cref{prob:key} using~\Cref{alg:base-case} as the base case of this family (i.e., $\alg_1$). 
Roughly speaking, the algorithm $\alg_k$ in this family achieves an update time $\approx \paren{m/n}^{1/(k+1)} \cdot \ORS{n}{\Theta_k(n)}^{1-1/(k+1)}$. 

\begin{lemma}\label{lem:recursive}
	There exists an absolute constant $c \geq 1$ such that the following holds. For any $k \geq 1$, there is an algorithm $\alg_k(n,m,q,\gamma,\delta,\alpha)$ (\Cref{alg:recursive}) for~\Cref{prob:key} that with high probability takes	
	\[
		O\Paren{(2q)^{k-1} \cdot \paren{\frac{m}{n}}^{1/k+1} \cdot  \ORS{n}{\gamma \cdot \delta n/2}^{1-1/(k+1)} \cdot n^{6\gamma} \cdot (\log{(n)}/\delta)^{c}},
	\]
	 amortized time over the updates to maintain the answer to all given queries.  The algorithm works as long as $\gamma < (1/12)^{k}$ and $\alpha \geq \gamma \cdot \delta$. 
\end{lemma}
 
We prove~\Cref{lem:recursive} by induction. The base case is handled by~\Cref{lem:base-case} for $\alg_1$. Now, suppose the lemma is true for some $k \geq 1$ and we 
prove it for $k+1$. The algorithm $\alg_{k+1}$ follows the same approach of~\Cref{alg:base-case} in processing the graph in batches of $t$ chunks of size $\alpha_{k+1} \cdot n_{k+1}$ for some $t$ to be determined later (it is going to be 
 $\approx \paren{{m}/{n}}^{k+1/(k+2)} \cdot \ORS{n}{\Theta(n)}^{1/(k+2)}$). In each batch, the algorithm also partitions the graph into three subgraphs $\Gold,\Gbatch,\Gmatch$ with
very similar definitions as in the past. 
The main difference however is that both graphs $\Gbatch$ and $\Gmatch$ are now handled by running algorithm $\alg_k$ over them (instead of the greedy approach of~\Cref{alg:base-case}). 

We first specify the parameters used for running $\alg_k$ on $\Gbatch$ and $\Gmatch$: 
\begin{alignat}{3}
	&n_k := n \qquad &&m_k := t \cdot q \cdot \alpha \cdot n \qquad &&q_k : = q \notag \\
	&\gamma_k := 12 \cdot \gamma \qquad &&\delta_k := \delta/12 \qquad &&\alpha_k := \alpha.  \label{eq:recursive-parameters} 
\end{alignat}

We will run $\alg_k(n_k,m_k,q_k,\gamma_k,\delta_k,\alpha_k)$ on $\Gbatch$ and $\Gmatch$. By the assumption in~\Cref{lem:recursive}, we have $\gamma < (1/12)^{k+1}$ and thus $\gamma_k = 12 \cdot \gamma < (1/12)^{k}$; hence we satisfy the condition for invoking the induction hypothesis of $\alg_k$. Similarly, we have $\alpha = \alpha_k$ and $\gamma_k \cdot \delta_k = \gamma \cdot \delta$ and so $\alpha_k \geq \gamma_k \cdot \delta_k$ also holds. 
Finally, at the beginning of each batch, $\Gbatch,\Gmatch$ are both empty graphs and thus satisfy the promise of~\Cref{prob:key}. As such, we can indeed apply the induction hypothesis to $\alg_k$ in the following (the only remaining part we need
to explicitly account for is to make sure that for each query $U$, we are guaranteed that $\mu(\Gbatch[U])$  is at least $\delta_k \cdot n$, before calling $\alg_k$ on $\Gbatch$ to be consistent with the definition of~\Cref{prob:key} (similarly for $\Gmatch$); we use~\Cref{prop:subtime-size} for ensuring this guarantee).  

The following algorithm follows the same strategy of~\Cref{alg:base-case} modulo applying $\alg_k$ to $\Gbatch$ and $\Gmatch$ instead of running the greedy algorithm over them (see also~\Cref{fig:schematic2}).

 \begin{Algorithm}[Algorithm $\alg_{k+1}(n,m,q,\gamma,\delta,\alpha)$ of~\Cref{lem:recursive}]\label{alg:recursive}
 	~
	\vspace{-5pt}
	\begin{itemize}
		\item Process the updates in batches $B$ of $t$ chunks $C_1,\ldots,C_t$ and for each batch: 
		\begin{enumerate}
			\item Let $\Gold = G$, $\Gbatch = \emptyset$, and $\Gmatch = \emptyset$ (on vertices $V$). Start two copies of $\alg_k(n_k,m_k,q_k,\gamma_k,\delta_k,\alpha_k)$ 
			on $\Gbatch$ and $\Gmatch$ separately with the parameters in~\Cref{eq:recursive-parameters}.  

			Maintain these graphs as follows for updates in each chunk (these graphs may also be updated based on queries): 
			\begin{enumerate}
			\item For an edge insertion $e=(u,v)$, add the edge to $\Gbatch$. 
			\item For an edge deletion $e=(u,v)$, remove the edge $e$ from any of the graphs $\Gold$, $\Gbatch$, or $\Gmatch$ that it belongs to currently. 
			\end{enumerate}
			\item After a chunk is updated, answer each query $U \subseteq V$ as follows: 
			\begin{enumerate}
				\item\label{line:recursive-batch} Run~\Cref{prop:subtime-size} on $\Gbatch[U]$ with parameter $\eps = (\delta_k/2)$ to obtain an estimate $\tmu:=\tmu(\Gbatch[U])$ of $\mu(\Gbatch[U])$. If $\tmu \geq \delta_k \cdot n$,
				pass the query $U$ to $\alg_k$ on $\Gbatch$ and return its output matching $\Mbatch$ as the answer; otherwise, continue.  
				\item\label{line:recursive-match} Run~\Cref{prop:subtime-size} on $\Gmatch[U]$ with parameter $\eps = (\delta_k/2)$ to obtain an estimate $\tmu:=\tmu(\Gmatch[U])$ of $\mu(\Gmatch[U])$. If $\tmu \geq \delta_k \cdot n$,
				pass the query $U$ to $\alg_k$ on $\Gmatch$ and return its output matching $\Mmatch$ as the answer; otherwise, continue. 
				\item\label{line:recursive-old} Run~\Cref{alg:subtime} on $\Gold$, the set $U$, and parameter $2\gamma$ with the guarantee that $\mu(\Gold[U]) \geq \delta n/2$ (which we establish in~\Cref{clm:recursive-correct}) 
				to obtain a matching $\Mold$. Remove $\Mold$ from $\Gold$ and insert it into $\Gmatch$.
			\end{enumerate}
		\end{enumerate}
	\end{itemize}
 \end{Algorithm}
 
 A remark about the updates in~\Cref{alg:recursive} is in order. Firstly, when processing an arriving chunk, to update $\Gbatch$ or $\Gmatch$, 
 we create two separate chunks of size $\alpha_k \cdot n_k$ based on these updates for $\alg_k$ on $\Gbatch$ and $\Gmatch$ separately (appending them with empty updates if needed, to have length exactly $\alpha_k \cdot n_k)$. 
We then pass these chunks to each algorithm as their updates. The updates to $\Gold$ are done directly. Moreover, in Line~\eqref{line:recursive-old}, we insert $\Mold$ of size $\gamma \cdot \delta  n \leq \alpha_k \cdot n_k$ (by the guarantee of~\Cref{lem:recursive} and choice of parameters in~\Cref{eq:recursive-parameters}) as a single chunk  (possibly with empty updates) to $\alg_k$ running on $\Gmatch$ (there will be no queries after these chunks for $\alg_k$). 

We first ensure that the subroutine calls in~\Cref{alg:recursive} are all valid. 

\begin{claim}\label{clm:recursive-calls-correct}
	With high probability, when running~\Cref{alg:recursive} on a single batch of $t$ chunks: 
	\begin{enumerate}
		\item $\Gbatch$ starts empty, at any point it has at most $m_k$ edges, and in total it receives $t$ chunks of size $\alpha_k \cdot n_k$ each; moreover, each query
		$U$ to $\alg_k$ on $\Gbatch$ satisfies $\mu(\Gbatch[U]) \geq \delta_k \cdot n$; 
		\item $\Gmatch$ starts empty, at any point it has at most $m_k$ edges, and in total it receives at most $t \cdot (q+1)$ chunks of size $\alpha_k \cdot n_k$; 
		moreover, each query $U$ to $\alg_k$ on $\Gmatch$ satisfies $\mu(\Gmatch[U]) \geq \delta_k \cdot n$; 
		\item $\Gold,\Gbatch,\Gmatch$ at any point partition the edges of $G$.
	\end{enumerate}
\end{claim}
\begin{proof}
	We prove each part as follows: 
	\begin{enumerate}
		\item At the beginning of the batch, we have $\Gbatch = \emptyset$ and for each chunk as input to~\Cref{alg:recursive}, $\alg_k$ also receives a chunk of size $\alpha_k \cdot n_k = \alpha \cdot n$ as an update (possibly with empty updates; 
		recall that $\Gbatch$ only processes insertions in the batch and deletions of the edges inserted during the batch). 
		Since there are $t$ chunks in each batch of~\Cref{alg:recursive}, there will be at most $t \cdot \alpha \cdot n$ edge insertions to $\Gbatch$ which is equal to $m_k/q \leq m_k$ by~\Cref{eq:recursive-parameters}. 
		Moreover, running~\Cref{prop:subtime-size} with high probability, returns $\tmu(\Gbatch[U]) \leq \mu(\Gbatch[U])$ and thus when~\Cref{alg:recursive} decides to query $\alg_k$ on $\Gbatch$, 
		we have $\mu(\Gbatch[U]) \geq \delta_k \cdot n$. 
		
		\item At the beginning of the batch $\Gbatch = \emptyset$ and for each chunk as input to~\Cref{alg:recursive}, $\alg_k$ also receives a chunk of size $\alpha_k \cdot n_k = \alpha \cdot n$ as an update, which
		can only include deletions and empty updates. Moreover, each time the algorithm reaches Line~\eqref{line:recursive-old}, it will be inserting edges of a matching $\Mold$ as updates in chunks of size $\alpha_k \cdot n_k$. 
		Given that $\Mold$ is of size $\gamma \cdot \delta n \leq \alpha n$ (by the technical assumption on $\gamma,\delta,\alpha$ in~\Cref{prob:key}), this translates into having one more chunk here as well. 
		Thus, for each chunk and each of its $q$ queries, we may insert another chunk of size $\alpha_k \cdot n$ into $\Gmatch$, implying that the total number of inserted chunks is $t \cdot (q+1)$. 
		Moreover, running~\Cref{prop:subtime-size} with high probability, returns $\tmu(\Gmatch[U]) \leq \mu(\Gmatch[U])$ and thus when~\Cref{alg:recursive} decides to query $\alg_k$ on $\Gmatch$, 
		we have $\mu(\Gmatch[U]) \geq \delta_k \cdot n$. 
		
		\item This step simply follows from the above (on validity of updating edges of $\Gbatch$ and $\Gmatch$) and since we only remove an edge from $\Gold$ if it is deleted or if it is moved to $\Gmatch$. \qed		
	\end{enumerate}
\end{proof}

 \begin{figure}[t!]
 \centering
 \includegraphics[scale=0.4]{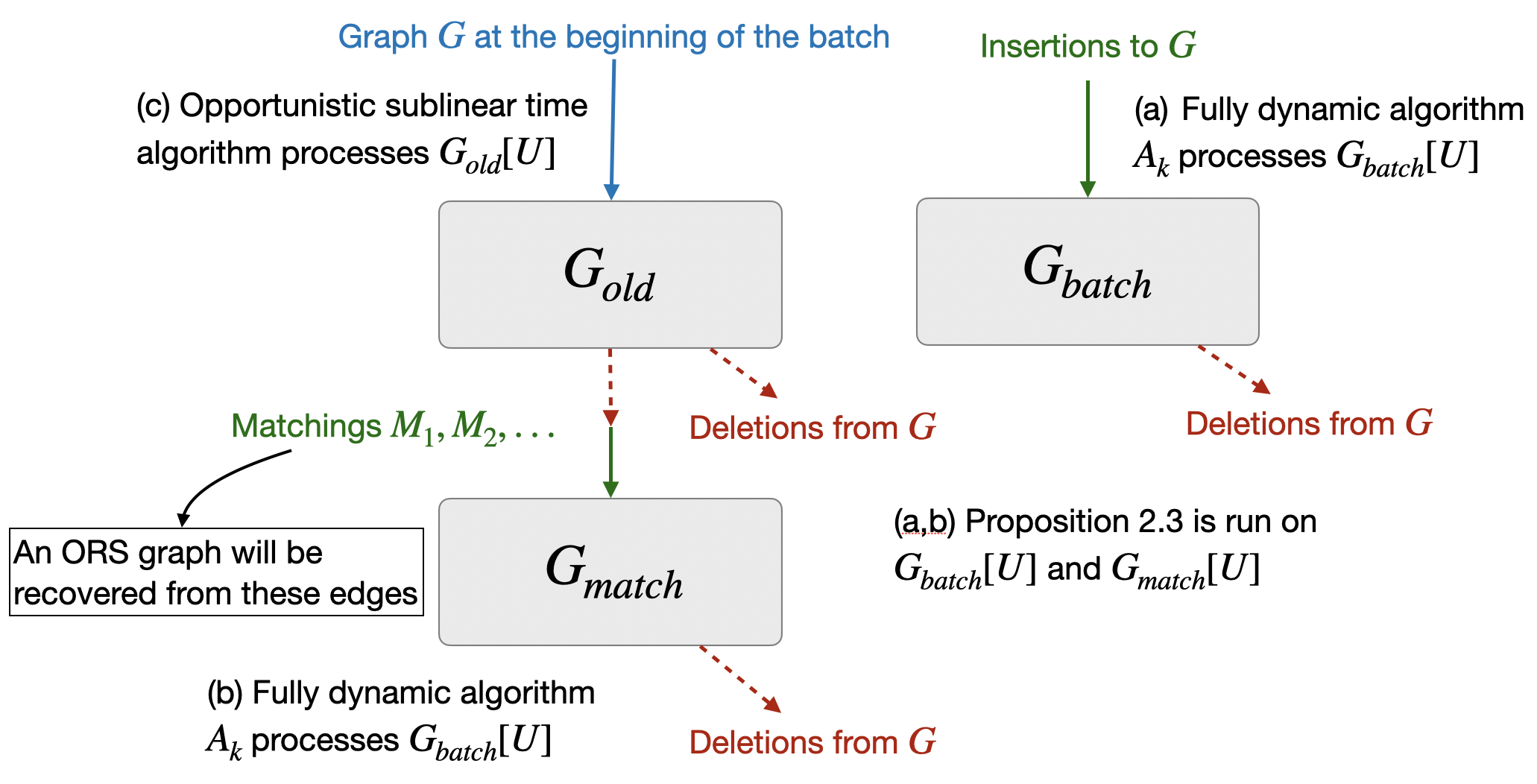}
 \caption{An illustration of the three graphs $\Gold,\Gbatch,\Gmatch$ in~\Cref{alg:recursive} for $\alg_{k+1}$, their role, and how they are being processed.
 Notice that $\Gold$ is a decremental graph, while $\Gbatch,\Gmatch$ are fully dynamic.  The main difference with~\Cref{alg:base-case} is that $\Gbatch$ and $\Gmatch$ are now 
 being handled recursively with $\alg_k$ (steps $(a)$ and $(b)$ also now involve running~\Cref{prop:subtime-size} to check if applying $\alg_k$ is valid). This algorithm also form an ORS from the 
 edges of the matchings $M_1,M_2,\ldots,$ moved from $\Gold$ to $\Gmatch$. }\label{fig:schematic2}
 \end{figure}

We can now establish the correctness of~\Cref{alg:recursive}. 

\begin{claim}\label{clm:recursive-correct}
 	With high probability, the answer to each query $U$ in \Cref{alg:recursive} is a valid answer according to~\Cref{prob:key}. 
 \end{claim}
 \begin{proof}
 	\Cref{clm:recursive-calls-correct} ensures that each query to recursive calls on $\alg_k$ on $\Gbatch$ and $\Gmatch$ returns a valid answer (given all promises required by~\Cref{prob:key} and the induction hypothesis of~\Cref{lem:recursive} for running this algorithm are satisfied). Thus, it remains to consider the case when~\Cref{alg:recursive} reaches Line~\eqref{line:recursive-old} to answer the query $U$. 
 
 	Recall that we have the assumption $\mu(G[U]) \geq \delta n$. At the same time, 
	each call to~\Cref{prop:subtime-size} in Lines~\eqref{line:recursive-batch} and~\eqref{line:recursive-match} guarantees that, respectively, 
	\begin{align*}
		\tmu(\Gbatch[U]) &\geq \frac{1}{2} \cdot \mu(\Gbatch[U]) - \frac{1}{2} \cdot \delta_k \cdot n,  \\
		\tmu(\Gmatch[U]) &\geq \frac{1}{2} \cdot \mu(\Gmatch[U]) - \frac{1}{2} \cdot \delta_k \cdot n.
	\end{align*}
 	Thus, if the algorithm has reached Line~\eqref{line:recursive-old}, we know that 
	\begin{align*}
		\mu(\Gbatch[U]) &\leq 2 \tmu(\Gbatch[U]) + \delta_k \cdot n \leq 3 \cdot \delta_k \cdot n, \\
		\mu(\Gmatch[U]) &\leq 2 \tmu(\Gmatch[U]) +  \delta_k \cdot n \leq 3 \cdot \delta_k \cdot n. 
	\end{align*}
	We further have that $\Gold,\Gbatch,\Gmatch$ at any point partition $G$ by~\Cref{clm:recursive-calls-correct}. Thus, 
	\[
		\mu(\Gold[U]) \geq \mu(G[U]) - \mu(\Gbatch[U]) - \mu(\Gmatch[U]) \geq \delta  n - 6 \cdot \delta_k n \geq \delta  n/2,
	\]
	by the choice of $\delta_k = \delta/12$ in~\Cref{eq:recursive-parameters}. 
	This implies that $\Gold[U]$ satisfies the guarantee of~\Cref{thm:subtime} for~\Cref{alg:subtime} with parameters $(\delta/2)$ and $(2\gamma)$. As such, this algorithm,
	with high probability, returns a matching $\Mold$ of size $\gamma \cdot \delta n$ from $\Gold[U]$, concluding the proof. 
 \end{proof}

The following claim is a direct analogue of~\Cref{clm:base-case-ORS}. Its proof is verbatim as before and hence is omitted here. 

\begin{claim}\label{clm:recursive-ORS}
	Let $M_1,M_2,\ldots,M_{\rho}$ be the matchings computed from $\Gold$ in Line~\eqref{line:recursive-old} of~\Cref{alg:recursive} and added to $\Gmatch$ \underline{at the time of their computation} (i.e., here, we ignore the deletions that have happened 
	subsequently, namely, some edges of $M_i$ might have been deleted from $\Gold$ when we are inserting $M_{i+1}$, but we still keep those edges in the definition of $M_i$).  
	These matchings are edge-disjoint and for every $i \in [\rho]$, maximum degree of $M_i$ among the matchings $M_{i},\ldots,M_{\rho}$ is at most $\indeg{M_i}$ in the graph $\Gold$ at the time $M_i$ was computed. 
\end{claim}

The last part is then to bound the runtime of the algorithm. 

\begin{claim}\label{clm:recursive-runtime}
	With high probability, when running~\Cref{alg:recursive} on a single batch of $t$ chunks: 
	\begin{enumerate}
		\item the total time spent for maintaining the graphs and bookkeeping is 
		\[
		O(t \cdot \alpha \cdot n);
		\] 
		\item the total time spent for running~\Cref{prop:subtime-size} in Lines~\eqref{line:recursive-batch} and~\eqref{line:recursive-match} is
		\[
		O(t \cdot q \cdot n \cdot \poly(\log{(n)}/\delta));
		\]
		\item the total time spent computing $\Mbatch$ in Line~\eqref{line:recursive-batch} is
		\[
		O\Paren{t \cdot \alpha \cdot n \cdot (2q)^{k-1} \cdot \Paren{t \cdot q \cdot \alpha}^{1/(k+1)} \cdot \ORS{n}{\gamma \cdot \delta n/2}^{1-1/(k+1)} \cdot n^{6\gamma} \cdot \log^{2}{\!(n)}};
		\]
		\item the total time spent computing $\Mmatch$ in Line~\eqref{line:recursive-match} is at most
		\[
			O\Paren{t \cdot (q+1) \cdot \alpha \cdot n \cdot (2q)^{k-1} \Paren{t \cdot q \cdot \alpha}^{1/(k+1)} \cdot \ORS{n}{\gamma \cdot \delta n/2}^{1-1/(k+1)} \cdot n^{6\gamma} \cdot \log^{2}{\!(n)}};
		\] 
		\item the total time spent computing $\Mold$ in Line~\eqref{line:recursive-old} is at most
		\[
		O\Paren{m \cdot \ORS{n}{\gamma \cdot \delta n/2} \cdot  n^{6\gamma} \cdot \log^2\!{(n)}}.
		\]
	\end{enumerate}
\end{claim}
\begin{proof}
	The proof just follows the same argument as in~\Cref{clm:base-case-runtime}. 
	
	Specifically, the first part follows immediately, and the second part is by~\Cref{prop:subtime-size}. 
	For parts three and four, plugging the choice of $m_k = t \cdot q \cdot \alpha \cdot n$ when applying the induction hypothesis of~\Cref{lem:recursive} for $\alg_k$ on $\Gbatch$ and $\Gmatch$, implies the bounds. 
	
	Finally, the last part holds by~\Cref{clm:recursive-ORS} and~\Cref{lem:ORS} exactly as in~\Cref{clm:base-case-runtime} as here also, each matching $\Mold$ is of size $\gamma \cdot \delta n$ and is chosen from $\Gold$ which 
	is a decremental graph throughout the batch. 
\end{proof}

\begin{proof}[Proof of~\Cref{lem:recursive}]
	The correctness of the algorithm follows from~\Cref{clm:recursive-correct} and a union bound over $\poly(n)$ intermediate graphs created in~\Cref{prob:key} (by the assumption on number of updates).  
	
	Furthermore, the runtime per each of $t \cdot \alpha \cdot n$ updates during a batch, by~\Cref{clm:recursive-runtime}, is at most 
	\begin{align*}
		&\hspace{23pt} O(t \cdot \alpha \cdot n)  \\
		&\hspace{10pt} + O(t \cdot q \cdot n \cdot \poly(\log{(n)}/\delta))  \\
		&\hspace{10pt} +O\Paren{t \cdot \alpha \cdot n \cdot (2q)^{k-1} \cdot \Paren{t \cdot q \cdot \alpha}^{1/(k+1)} \cdot \ORS{n}{\gamma \cdot \delta n/2}^{1-1/(k+1)} \cdot n^{6\gamma} \cdot \log^{2}{\!(n)}}  \\
		&\hspace{10pt} +O\Paren{t \cdot (q+1) \cdot \alpha \cdot n \cdot (2q)^{k-1} \Paren{t \cdot q \cdot \alpha}^{1/(k+1)} \cdot \ORS{n}{\gamma \cdot \delta n/2}^{1-1/(k+1)} \cdot n^{6\gamma} \cdot \log^{2}{\!(n)}}  \\
		&\hspace{10pt} + O\Paren{m  \cdot \ORS{n}{\gamma \cdot \delta n/2} \cdot n^{6\gamma} \cdot \log^2\!{(n)}} \\
		&\hspace{10pt}= O\Paren{t \cdot (2q) \cdot \alpha \cdot n \cdot (2q)^{k-1} \Paren{t \cdot (2q) \cdot \alpha}^{1/(k+1)} \cdot \ORS{n}{\gamma \cdot \delta n/2}^{1-1/(k+1)} \cdot n^{6\gamma} \cdot (\log(n)/\delta)^{c}} \\
		&\hspace{10pt} + O\Paren{m \cdot \ORS{n}{\gamma \cdot \delta n/2} \cdot n^{6\gamma} \cdot (\log(n)/\delta)^{c}},
	\end{align*}
	where in the equality, we used several loose upper bounds (to simplify the subsequent calculations) and use $c$ as the absolute constant which is equal to the exponent of the $\poly$-term in~\Cref{prop:subtime-size} (we also take $c > 2$ to  
	subsume the $\log^{2}(n)$ term of prior equations). 
	
	We can now balance these terms by setting 
	\[
		t :=  \paren{\frac{m}{n}}^{k+1/(k+2)} \cdot \ORS{n}{\gamma \cdot \delta n/2}^{1/(k+2)} \cdot \frac{1}{(2q)^{\frac{(k-1) \cdot (k+1) + (k+2)}{k+2}} \cdot \alpha}, 
	\]
	which leads to the amortized update time of
	\begin{align*}
		&O\Paren{m \cdot \ORS{n}{\gamma \cdot \delta n/2} \cdot n^{6\gamma} \cdot (\log(n)/\delta)^{c} \cdot \frac{1}{t \cdot \alpha \cdot n}} \\
		&\hspace{1cm}= O\Paren{ \paren{\frac{m}{n}}^{1/(k+2)} \cdot \ORS{n}{\gamma \cdot \delta n/2}^{1-1/(k+2)} \cdot (2q)^{k} \cdot n^{6\gamma} \cdot (\log(n)/\delta)^{c}}, 
	\end{align*}
	where we used $(k-1) \cdot (k+1) + (k+2) = k^2 + k + 1 \leq k \cdot (k+2)$ for $k \geq 1$.
	
	Similar to the proof of~\Cref{lem:base-case}, we should also handle the case wherein the total number of chunks given to the algorithm 
	does not even reach a single batch. As before, in this case, given the promise that the graph $G$ starts empty, the only graph that is non-empty will be $\Gbatch$, 
	and thus the amortized runtime of the algorithm is the same as $\alg_k$ on $\Gbatch$ (with the given parameters, in particular $m_k = t \cdot q \cdot \alpha \cdot n$). 
	Thus, the amortized runtime of~\Cref{alg:recursive} will be at most (by~\Cref{clm:recursive-runtime} for bookkeeping, running~\Cref{prop:subtime-size}\footnote{In fact, for the first batch, we do not even need to run~\Cref{prop:subtime-size}, 
	given that we know $\mu(G[U]) \geq \delta n$ (by the promise of~\Cref{prob:key} and since we know $\Gold = \Gmatch=\emptyset$. We ignore this extra optimization step since it does not affect the overall runtime of the algorithm.}%
	and running $\alg_k$ on $\Gbatch$), 
	\begin{align*}
		&O\Paren{(2q)^{k-1} \cdot \Paren{t \cdot q \cdot \alpha}^{1/(k+1)} \cdot \ORS{n}{\gamma \cdot \delta n/2}^{1-1/(k+1)} \cdot n^{6\gamma} \cdot (\log(n)/\delta)^{c}} \\
		&\hspace{1cm}= O\Paren{ \paren{\frac{m}{n}}^{1/(k+2)} \cdot \ORS{n}{\gamma \cdot \delta n/2}^{1-1/(k+2)} \cdot (2q)^{k-1} \cdot n^{6\gamma} \cdot (\log(n)/\delta)^{c}}, 
	\end{align*}
	by the choice of $t$ (using the same exact calculation and the above step). 
	
	 This proves the induction step of~\Cref{lem:recursive} and concludes the proof. 
\end{proof}


\section{A Fully Dynamic Algorithm for Maximum Matching}\label{sec:dynamic}

The following theorem, which is the main contribution of our work, formalizes~\Cref{res:main}. 

\begin{theorem}\label{thm:main}
	Let $\eps \in (0,1/100)$ be a given parameter and $k \geq 1$ be any integer. Let $\gamma = (1/20)^{k}$ and $f(\gamma,\eps/4)$ and $g(\gamma,\eps/4)$ be as defined in~\Cref{prop:boosting}. 
	
	There exists an algorithm for maintaining a $(1-\eps)$-approximation to maximum matching
	in a fully dynamic $n$-vertex graph that starts empty with amortized update time of
	\[
			O\Paren{{n}^{1/k+1} \cdot  \ORS{n\,}{\,\frac{1}{15^k} \cdot f(\gamma,\Theta(\eps)^2) \cdot n}^{1-1/(k+1)} \cdot n^{15/(20)^{k}}}.
	\]	
	The guarantees of this algorithm hold with high probability even against an adaptive adversary. 
\end{theorem}


\begin{proof}[Proof of~\Cref{thm:main}]
	The proof is a combination of the standard tools listed in~\Cref{sec:tools} to reduce the problem to~\Cref{prob:key} and 
	then applying~\Cref{lem:recursive} for solving this problem. 
	
	We start by obtaining an algorithm for an \emph{additive} $\eps \cdot n$ approximation to maximum matching. 
	Define the following parameters for solving~\Cref{prob:key}
	via the algorithm $\alg_k(n_k,m_k,q_k,\gamma_k,\delta_k,\alpha_k)$: 
	\begin{alignat}{3}
	&n_k := n \qquad &&m_k := \binom{n}{2} \qquad &&q_k : = g(\gamma,\eps/4) \notag \\
	&\gamma_k := (1/15)^k \qquad &&\delta_k := f(\gamma,\eps/4) \qquad &&\alpha_k := \eps^2.  \label{eq:dynamic-parameters} 
	\end{alignat}
	Suppose we have computed a $(\eps/4) \cdot n$ additive approximate matching at some time. Then, for the next $\alpha_k \cdot n = \eps^2 \cdot n$ this remains at least a $(\eps/2) \cdot n$ approximation (even if all updates delete edges of this matching). 
	At this point, we should recompute another $(\eps/4) \cdot n$ additive approximate matching. By~\Cref{prop:boosting}, at this point, we need to answer $q_k := g(\gamma,\eps/4)$ queries $U \subseteq V$ 
	with $\mu(G[U]) \geq \delta_kn = f(\gamma,\eps/4) \cdot n$ and returning a matching of size $\mu(G[U]) \geq \gamma_k \cdot \delta_k \cdot n$ satisfies the requirement of answering the queries in~\Cref{prop:boosting}. 
	However, we do need to ensure that $\mu(G[U]) \geq \delta_k \cdot n$ which can be done by running~\Cref{prop:subtime-size} first. At this point, the problem we need to solve to implement~\Cref{prop:boosting} 
	is exactly~\Cref{prob:key} with the parameters of~\Cref{eq:dynamic-parameters}.  
	
	Running~$\alg_k$ for solving~\Cref{prob:key} by~\Cref{lem:recursive} is going to have an amortized update time of 
	\begin{align*}
		&O\Paren{(2q_k)^{k-1} \cdot \paren{\frac{m_k}{n_k}}^{1/k+1} \cdot  \ORS{n_k}{\gamma_k \cdot \delta_k \cdot n/2}^{1-1/(k+1)} \cdot n^{6\gamma_k} \cdot (\log{(n)}/\delta_k)^{c}} \\
		&\hspace{10pt}=O\Paren{{n}^{1/k+1} \cdot  \ORS{n\,}{\,\frac{1}{20^k} \cdot f(\gamma,\eps/4) \cdot n}^{1-1/(k+1)} \cdot n^{10/(20)^k}}, 
	\end{align*}
	with high probability; here, we used that $(2g(\gamma,\eps/4))^{k-1} \cdot n^{6/(20)^k} \cdot (\log{(n)}/\delta_k)^{c} = O(n^{10/(20)^k})$ given the values $f(\gamma,\eps),g(\gamma,\eps),\delta_k = O_{k,\eps}(1)$. 
	Also,  the runtime of $O(n \cdot \poly(\log{(n)}/\delta_k))$ for running~\Cref{prop:subtime-size} amortized over the $\eps^2 \cdot n$ updates is asymptotically upper bounded by the above and can be neglected. 
	All in all, this implies an algorithm with amortized update time of
	\[
		O\Paren{{n}^{1/k+1} \cdot  \ORS{n\,}{\,\frac{1}{20^k} \cdot f(\gamma,\eps/4) \cdot n}^{1-1/(k+1)} \cdot n^{10/(20)^k}},
	\]
	for maintaining an additive $\eps n$ approximate matching in a dynamic graph, with high probability. 
		
	Finally, we apply~\Cref{prop:additive} to turn this into a multiplicative $(1-\eps)$-approximation guarantee. 
	This effectively requires re-parameterizing $\eps$ with $\Theta(\eps^2)$ in the above bounds (and bounding $n^{10/(20)^k} \cdot \poly(\log{(n)}/\eps) = O(n^{15/(20)^k})$ by the range of parameters). 
	This implies an algorithm with amortized update time promised in the theorem statement 
	for maintaining a (multiplicative) $(1-\eps)$-approximation to maximum matching in a fully dynamic graph, with high probability. 
	
	We shall remark that in the arguments above---in particular, to satisfy the guarantee promised in~\Cref{prob:key}---we need to assume that the total number
	of updates is $\poly(n)$ to be able to apply union bound in conjunction with our high probability guarantees. This however can be easily fixed using a standard trick\footnote{In fact, it appears that
	many of existing dynamic matching algorithms make this assumption implicitly, e.g., in~\cite{BhattacharyaKS23,Liu24,BehnezhadG24}, although some are also more explicit about this, e.g.~\cite{BehnezhadK22}.}
	as we explain next.
	
	After every, say, $n^{10}$ updates to the underlying graph $G$, we entirely terminate the current run of the algorithm and erase all the data structures. Then, we start a new run of the algorithm on an initially empty graph $H$ and insert the current edges in $G$ 
	to this new graph $H$ using at most $O(n^2)$ insertions. After this step, the graph $H$ becomes the same as $G$ and we continue with processing the upcoming updates the current data structures we have. This does not change the 
	asymptotic runtime of the algorithm, but now ensures that even if the algorithm needs to process super-polynomial number of updates, after \emph{each} update, with high probability, the output is correct
	and the amortized runtime of the algorithm is as desired\footnote{We shall emphasize that this does not mean on such a long sequence, the guarantees are satisfied \emph{throughout} (as we simply cannot do a union bound over 
	so many events). However, even an adaptive adversary cannot make sure that a \emph{fixed} update returns a wrong answer or take longer than guaranteed except with negligible probability.}. This concludes the proof. 
\end{proof}

\subsection{Removing the Assumption on the Prior Knowledge of ORS} 

{In the description of our algorithms throughout this paper, we assumed that the algorithm is aware of the value of $\ORS{n}{c \cdot n}$ (for a proper choice of $c$ in the algorithm) to find the right balancing point for size of batches. However, as we explain next, 
this knowledge is actually not necessary, which is a desirable feature given our current state of (lack of) understanding of $\ORS{n}{c \cdot n}$.

We again focus on the case of polynomially many updates. The algorithm starts with \emph{guessing} that $\ORS{n}{c \cdot n}$ is some $\beta =n^{\Omega(1/\log\log{n})}$ (the current best lower bound in~\Cref{eq:ORS}), runs the algorithm of~\Cref{thm:main} with $\beta$ as the value of ORS, and continues as long as the runtime does not exceed the bounds dictated by the current guess $\beta$. Now suppose during some run of the algorithm, the runtime exceeds the current update time bound. Then there are two possibilities: either the high probability event of~\Cref{thm:main} has failed or the guess $\beta$ of  $\ORS{n}{c \cdot n}$ falls shorts of the true value. The first case happens with a negligible probability which we can ignore 
so let us focus on the second case. 

The guarantees given in~\Cref{lem:recursive} hold for each fixed batch of updates (in other words, we amortize the runtime over a single batch and not beyond that). 
Suppose we have a batch with longer than expected runtime. We consider all the matchings moved from $\Gold$ to $\Gmatch$ 
throughout this batch and apply~\Cref{lem:ORS} to them in an algorithmic fashion---which takes linear time in the size of the graph---to explicitly construct an ORS graph. If 
we succeed in creating an ORS graph with strictly more matchings than the current estimate $\beta$,
we terminate the current algorithm, increase the value of $\beta$ by a factor of $4$, and restart the process from the beginning of the last batch. Since for any $k \ge 1$, the target runtime of the algorithm $\alg_k$ is proportional to $\ORS{n}{c \, n}^{\eta}$ for some $\eta \ge 1/2$, running the algorithm again with $\beta$ increased by a factor of $4$, results in a geometrically increasing sequence of runtimes. This makes the final runtime only a constant factor larger than if the algorithm had been run with the correct value of $\ORS{n}{c \cdot n}$ from the very beginning. We also note that the value of $\beta$ only monotonically increases over the execution of the algorithm because at every occurrence, we recover a certificate of a new lower bound on the value of $\ORS{n}{c \cdot n}$. So the revision of parameter $\beta$ occurs only $O(\log n)$ times over the {\em entire} execution of the algorithm.


In summary, we are able to recover the bounds of~\Cref{thm:main} without a prior knowledge of the value of $\ORS{n}{c \cdot n}$ as also advertised in~\Cref{res:main}.
}

\section*{Acknowledgement} 
\addcontentsline{toc}{section}{Acknowledgement}

We would like to thank Soheil Behnezhad and Alma Ghafari for helpful discussions on their results in~\cite{BehnezhadG24}, and 
Aaron Bernstein, Sayan Bhattacharya, and Thatchaphol Saranurak for their shared discussions and insights on this problem. 

Part of this work was conducted while the first named author was visiting the Simons Institute for the Theory of Computing as part of the Sublinear Algorithms program. 

\clearpage

\bibliographystyle{halpha-abbrv}
\bibliography{general}



\end{document}